%% file: arxiv.tex
\documentclass[acmsmall, nonacm]{acmart}

\input{preamble.tex}

\input{macros.tex}

\title{Early-Stabilizing Counting}

\author{Christoph Lenzen}
\affiliation{%
  \institution{Aalto University}
  \city{Helsinki}
  \country{Finland}
}
\affiliation{%
  \institution{Reykjavik University}
  \city{Reykjavik}
  \country{Iceland}
}
\email{christophlen@ru.is}

\author{Julian Loss}
\affiliation{%
  \institution{Ruhr University Bochum}
  \city{Bochum}
  \country{Germany}
}
\email{julian.loss@rub.de}

\begin{document}

\SetKwInput{KwVars}{Variables} 

\begin{abstract}
  \input{abstract.tex}
\end{abstract}

\maketitle

\input{intro.tex}
\input{prelim.tex}
\input{overview.tex}

\input{conventions.tex}
\input{counting.tex}
\input{communication.tex}
\input{opt_stab.tex}

\input{open.tex}

\bibliographystyle{plain}
\bibliography{references}

\appendix
\input{app_implicit.tex}

\end{document}

%% file: preamble.tex
\usepackage{amsthm,thmtools,thm-restate}
\usepackage{graphicx}

\usepackage[algo2e,noend,ruled,linesnumbered]{algorithm2e}
\DontPrintSemicolon 
\usepackage{silence}
\usepackage{hyperref} 
\hypersetup{
	colorlinks=true,  
	linkcolor=RoyalBlue,   
	filecolor=magenta, 
	urlcolor=RoyalBlue,    
	citecolor=OliveGreen,  
	hypertexnames=false
}
\usepackage[capitalize]{cleveref}

\usepackage{enumitem}


%% file: macros.tex
\newcommand{\NN}{\mathbb N}

\newcommand{\ZZ}{\mathbb Z}

\newlength{\protowidth}

%% file: abstract.tex
Synchronous \emph{Counting} is the task of reaching agreement on a common round counter in a synchronous system of $n$ nodes with up to $t$ Byzantine faults in a \emph{self-stabilizing} manner. That is, after transient faults may have arbitrarily corrupted the system state and ceased, the at least $n-t$ non-faulty nodes need to (re-)establish that (i) their local outputs are identical and (ii) increase by $1$ modulo $C$ in each round. An overhead-free reduction from consensus shows that all known lower bounds and impossibilities for consensus carry over to the counting problem. In the other direction, prior work has established that a consensus algorithm $\mathcal{A}$ can be turned into a counting algorithm at small overhead relative to the running time and bit complexity of $\mathcal{A}$, without losing resilience.

Taking inspiration from early-stopping consensus protocols, in this work we introduce the concept of \emph{early stabilization.} That is, if there are $0\le f\le t$ (persistent) faults in an execution, the algorithm should stabilize in a number of rounds that depends on $f$ only. Likewise, we seek to achieve an amortized bit complexity that is \emph{adaptive} in the number of actual faults $f$. By developing a number of modular building blocks suitable to these goals, we develop a $C$-counting algorithm that stabilizes within asymptotically optimal $O(f+1)$ rounds, has message size $O(\log^2 n + \log C)$, and has amortized bit complexity $O(n(f\log C +\log^2 n))$.

%% file: intro.tex
\section{Introduction \& Related Work}\label{sec:intro}
Synchronous \emph{Counting,} a.k.a.\ Byzantine digital clock synchronization, asks the nodes of a synchronous system to agree on a common round counter despite arbitrary initial states and interference from Byzantine faulty parties. This task naturally arises once its big brother, Byzantine fault-tolerant pulse synchronization~\cite{dolev04self,lenzen19easy}, has been solved: \emph{this} problem establishes synchronized, self-stabilizing regular pulses in a system where nodes are equipped with local clocks of bounded relative drift and exchange messages whose communication delay varies within known bounds. These pulses can serve as a clock signal establishing the synchronous abstraction despite both bounded Byzantine and and unbounded transient faults, but the resulting rounds are ``anonymous.'' This gets in the way of basic operations like initiating a subroutine that should be executed regularly, but not \emph{every} round, or responding coherently to external signals that are not necessarily observed by all nodes in the same round. The latter can be achieved by solving the closely related firing squad problem, cf.~\cite{lenzen19counting}.

Thus, at first glance, counting is a fundamental challenge with practical utility. However, the latter demands a more nuanced perspective. While applying pulse synchronization naively results in unlabeled rounds, one can use them as a ``heartbeat'' to stabilize a faster, more accurate pulse synchronization algorithm that would not stabilize on its own, but labels pulses modulo $C$~\cite{khanchandani19precision}. This comes at the expense of slowing down the heartbeat by a factor of $\Theta(C)$, which negatively impacts the stabilization time of the overall protocol stack. Ben-Or et al.~\cite{benor08fast} achieve a much smaller constant expected time overhead, by leveraging a self-stabilizing stream of weak shared coins. However, this comes at the expense of large communication overhead and additional assumptions.\footnote{The classic protocol by Feldman and Micali would require $\Omega(n^3\kappa)$ bits \emph{per heartbeat,} where $\kappa$ is the security parameter, and is secure against a static, computationally bounded adversary. One cannot use more efficient protocols that rely on a trusted setup, as transient faults could reveal any pre-shared secret information.} Alternatively, one can build large clocks from smaller ones~\cite{lenzen13socs}, turning the multiplicative overhead in stabilization time of~\cite{khanchandani19precision} into an additive one. This work gives a deterministic algorithm that runs in $O(t)$ rounds and tolerates $t<n/3$ Byzantine faults, alongside a randomized solution that achieves stabilization within $O(t+\log C)$ rounds with probability $1-2^{-t+\log C}$. These bounds are optimal up to small factors in the worst case~\cite{chor89simple,fischer82time,pease80reaching} due to a folklore reduction from consensus to counting. A remaining niche case is when the goal is to stabilize within $S\in \log^{O(1)} n$ rounds. Here, combining \cite{lenzen19counting} and \cite{lenzen19easy} achieves this property with probability $1-2^{-\log^{O(1)} n}$, at the expense of higher communication complexity (depending on the consensus algorithm plugged into their framework). This can be viewed a generalization of the aforementioned earlier work by Ben-Or et al.~\cite{benor08fast} that also solves the problem of generating the heartbeats, without needing to rely on a shared coin primitive.

So, should we cast aside the counting problem or at least reduce it to the status of a theoretical curiosity? We argue that this would be premature, as the above characterization limits its view to assuming a worst-case number of faults. The round lower bound for consensus breaks down when considering executions with $f\ll t$ faults~\cite{dolev82eventual}, or more precisely, requires a more fine-grained argument that yields a round complexity of $\min\{f+2,t+1\}$~\cite{dolev90early}. Regarding bit complexity, current solutions all send $\Omega(nt)$ bits \emph{per round,} factor $t$ beyond what the strongest known lower bound for consensus implies for counting~\cite{dolev85bounds,abraham19revisited}. Accordingly, in this work we explore the question how much can be gained by optimizing stabilization time and bit complexity as function of $f$ rather than $t$.

\subsection{Our Contributions}
To the best of our knowledge, this work is the first to carry over the concept of early-stopping consensus algorithms, which terminate within $O(f+1)$ rounds, to the realm of self-stabilization. We coin algorithms in which the stabilization time depends only on $f$, i.e., $S=S(f)$, \emph{early-stabilizing.} Likewise, we are unaware of prior attempts to reduce the communication cost of fault-tolerant self-stabilizing algorithms in executions with few faults. On a high level, we perceive three key benefits in doing~so.
\begin{enumerate}
  \item Reducing communication complexity simply saves resources. While it is true that larger bandwidth still needs to be available \emph{if} there is a large number of faults, these might not simultaneously affect all parts of the system or network. Hence, if communication resources are shared with other subroutines or subsystems, it can even be possible that overall bandwidth requirements are reduced for achieving the same level of resilience.
  \item Reduced stabilization time can clearly improve system availability. However, in the context of self-stabilization, we emphasize a compounding impact that goes way beyond shorter recovery time after a burst of transient faults overwhelms the system. Quicker recovery affects also how quickly nodes undergoing transient faults reattain a state consistent with the system when the fault threshold $t$ is not reached. This entails that they contribute again to the quorum of $n-t$ synchronized nodes required to count despite Byzantine faults. For example, if $t=\lfloor (n-1)/3\rfloor$ and in each round each node suffers a transient fault with independent probability $9/n$ (and there are no other faults), straightforward calculations and concentration bounds show the following.
\begin{itemize}
  \item [(i)] Within $t$ rounds, with probability $1-2^{-\Omega(n)}$, more than $t$ distinct nodes suffer a transient fault. Hence, a system in which $t+1$ or more rounds are required for individual nodes to recover from transients would be virtually guaranteed to fail globally within $t$ rounds.\footnote{For this informal line of reasoning, we assume that a transient fault will invalidate a node's state and it will take the full stabilization time to recover a consistent state.}
  \item [(ii)] Within $t/3$ rounds, the probability that more than $t$ nodes fail \emph{or} that during each subinterval of $S(0)\in O(1)$ rounds some transient fault occurs is $2^{-\Omega(n)}$. Thus, an algorithm that stabilizes in $S(0)$ rounds when there are no faults is guaranteed to maintain count for $2^{\Omega(n)}$ rounds with probability $1-2^{-\Omega(n)}$.
\end{itemize}
Put simply, early stabilization fundamentally improves the ability of the system to contain transient faults.
  \item In~\cite{lenzen19easy}, the authors mention that their work on pulse synchronization builds on earlier results on counting, where key techniques were developed in the ``sandbox'' provided by the more structured synchronous model. We hope the techniques we develop here for counting can be utilized in a similar way. 
\end{enumerate}
Concretely, we craft a set of modular tools that ultimately achieve the following main result; we fix optimal $t:=\lfloor (n-1)/3\rfloor$ from hereon.
\begin{restatable}{theorem}{main}\label{thm:main}
There is a $C$-counting algorithm with stabilization time $O(f+1)$, message size $O(\log^2 n + \log C)$, and bit complexity of $O(n(f\log C+\log^2 n))$ amortized over $n$ rounds.
\end{restatable}
Compared to the state of the art~\cite{lenzen13socs,lenzen19counting}, this reduces the stabilization time from being larger than $t$ to asymptotically optimal $O(f+1)$ \emph{and} the bit complexity from $\Omega(n^2)$ to (amortized) $\tilde{O}(n(f+1))$ per round.


\subsection*{Paper Organization}
In~\Cref{sec:model}, we discuss the system model, formalize the counting problem, and introduce two tasks that are highly useful in modularizing our solutions. Next, in~\Cref{sec:overview} we provide an overview of the key challenges in obtaining our results and sketch the ideas underlying our algorithms and proofs. To prepare for the technical exposition, in~\Cref{sec:conventions} we introduce a number of conventions we use in pseudocode to hide important, but distracting book-keeping operations and other details.\footnote{Readers familiar with self-stabilization should not be surprised by any of them, and a conceptual understanding does not require examining them closely. However, they imbue the pseudocode with a sufficiently rigorous meaning to enable proofs.} \Cref{sec:counting} discusses how to achieve stabilization time $O(f+\log n)$, focusing on the core idea in obtaining a stabilization time dominated by $f$ instead of $t$. We follow up by reducing amortized bit complexity to $\tilde{O}(n(f+1))$ in~\Cref{sec:communication}. To arrive at our main result, \Cref{sec:opt_stab} introduces additional building blocks and modifies the recursion template to remove the additive $O(\log n)$ overhead from the stabilization time.
We conclude the paper with a discussion of open questions in~\Cref{sec:open}.

%% file: prelim.tex
\section{Model and Preliminaries}\label{sec:model}
\paragraph{Communication Model}


We consider a fully connected synchronous system with node set $V:=[n]$, where we use the shorthand $[k]:=\{0,\ldots,k-1\}$. That is, algorithms proceed in \emph{rounds,} in which nodes send and receive messages and perform computations in lock-step. When a node $v$ sends a message $m$ to another node $w$ at the onset of a round $r\in\NN$, $w$ receives and processes $m$ by the end of round $r$, knowing that $m$ was sent by $v$. The \emph{round complexity} of a (non-stabilizing) protocol is the maximum number of rounds until all nodes have decided on their output and terminated, i.e., cease to send messages.

\paragraph{Fault Model}
In this work, we seek so-called \emph{self-stabilizing} solutions, i.e., the system recovers correct operation after a period of unbounded transient faults. We assume that by the start of the first round (that we analyze), the most recent such period is over. In addition, we require that the system can sustain any number of Byzantine \emph{faulty} nodes $f\le t:=\lfloor (n-1)/3\rfloor$, where recovery should succeed despite these ongoing faults. Accordingly, the non-Byzantine nodes, which we refer to as \emph{correct,} may initially have an arbitrary state, but the model assumptions on the communication network and the number of correct nodes hold again at the beginning of round $1$.\footnote{In particular, the code of the algorithm must be protected against alterations by transient faults, e.g.\ by being stored in more resilient non-volatile memory.} The goal is for the system as a whole to recover correct operation again within a bounded number of rounds $S$. This bound is referred to as \emph{stabilization time.} Note that self-stabilizing algorithms can never terminate; the stabilization time is their closest equivalent to round complexity~\cite{awerbuch91checking}.

We consider deterministic algorithms, i.e., the system must stabilize for any initial states of correct nodes and behavior of the faulty nodes. This is equivalent to assuming a powerful adversary that controls initial states and has perfect knowledge of the system, as the initial states and messages sent by faulty nodes fully determine an execution. In contrast, the algorithm is not aware of $f$, and must achieve its guarantees based on the knowledge that $f\le t$ only. This is crucial, as in this work we are interested in a stabilization time that depends on $f$, not $t$ or $n$: we set out to achieve $S\in O(f+1)$. In analogy to the concept of early-stopping algorithms~\cite{dolev82eventual}, we refer to this as \emph{early stabilization.}

As a secondary optimization criterion, we seek to minimize the amortized bit complexity of our algorithms. That is, if honest parties send a total of $B(r)$ bits in round $r$, there should be a value $R$ such that $(\sum_{r=1}^R B(r))/R\le B$. In this case, the algorithm has \emph{bit complexity of $B$ amortized over $R$ rounds.} We remark that it is common to ignore any communication by faulty nodes altogether, as they could send anything. However, our algorithms use short messages of polylogarithmic length, so that correct nodes can safely ignore longer messages without processing them, i.e., no backdoor through which an attacker could overtax correct nodes' computational capacity is introduced.

Finally, our self-stabilizing algorithms will list the variables that are \emph{maintained} between rounds, where we use subscript $v$ to indicate the node. Note that subroutines may come with their own variables, but we will use their outputs only. For analysis purposes, we will refer to the state of a variable $X_v$ of node $v$ \emph{at the end} of round $r\in \NN$ by $X_{v,r}$, where $X_{v,0}$ refers to the (arbitrary) initial state of $X_v$. We do the same for output variables of self-stabilizing subroutines.

\subsection*{Definitions}
The core task we consider in this work is digital clock synchronization or synchronous \emph{counting}.
\begin{definition}[$C$-Counting]\label{def:counting}
Suppose that each node maintains $C_v\in [C]$, where $2\le C\in \NN$. We say that these variables \emph{$C$-count} iff they satisfy for all correct nodes $v$ and $w$ and $r\in \NN$ that
\begin{itemize}
\item \textbf{Agreement}: $C_{v,r}=C_{w,r}$.
\item \textbf{Validity}: $C_{v,r+1}=C_{v,r}+1\bmod C$.
\end{itemize}
We say that the variables $C_v$ \emph{$C$-count from round $r\in \NN$} iff the above conditions hold in rounds $r'\ge r$.
\end{definition}
We formalize a task that is a core step in the construction used in~\cite{lenzen19counting}, which we will need to implement in a different way.
Intuitively, filtering distributes a $C$-counter that is (recursively) generated by a subset $T$ of the nodes to all, while limiting the inconsistency in views it the subset contains too many faulty nodes to correctly generate and distribute its shared count.
More concretely, the requirement is that over any $X$ consecutive rounds, there is only a single count value that increases by one modulo $C$ in each round that may be output by correct nodes;
the only other feasible output is $\bot$, a special symbol indicating (possibly transient) faults.
\begin{definition}[$(C,X,T)$-Clock Filtering]\label{def:filterclock}
Suppose that $v\in V$ maintains a variable $F_v\in [C]\cup \{\bot\}$, where $2\le C\in \NN$. Moreover, there is a designated \emph{clock set} $T\subset V$ that maintains variables $C_v\in [C]$. The variables $F_{v}$ \emph{$(C,X,T)$-filter} $C_v$ from round $r\in \NN$ iff the following properties are true:
\begin{itemize}
\item\textbf{Validity}: If fewer than $|T|/2$ nodes in $T$ are faulty and the variables $C_v$ $C$-count, then the variables $F_v$ $C$-count from round $r$.
\item\textbf{$X$-round Crusader Agreement}: If for $r'\ge r$ and correct $v$ it holds that $F_{v,r'}=c\bmod C$, then for all correct $w$ and rounds $\hat{r}\in\{r',\ldots r'+X\}$, it holds that $F_{w,\hat{r}}\in \{F_{v,\hat{r}}+\hat{r}-r'\bmod C,\bot\}$.
\end{itemize}
\end{definition}
For the sake of notational convenience, we adopt the following conventions with respect to the above definitions.
First, if $C$, $X$, or $T$ are clear from context, we may omit them.
Second, when we say that an algorithm solves one of these problems with stabilization time $S$, we mean that it maintains variables that count or filter from round $S$, respectively.

\begin{restatable}[implicit in \cite{lenzen17counting,lenzen19counting}]{lemma}{filtering}\label{lem:filtering}
$(C,X,T)$-Clock Filtering can be solved with stabilization time $S=X+2$, where each node sends an $O(\log |C|)$-sized message to each other node in each round.
\end{restatable}

A well-known building block of consensus algorithms is graded agreement, which will come in handy as a subroutine for us as well.
\begin{definition}[Graded Agreement]
The \emph{Graded Agreement} problem is specified as follows.
Each node $v$ has input $x_v\in \mathcal{V}$ and computes an output $(y_v,g_v)\in \mathcal{V}\times \{0,1\}$ with the following guarantees:
\begin{itemize}
  \item \textbf{Validity}: If there is $x\in \mathcal{V}$ so that $x_v=x$ for all correct $v$, then $(y_v,g_v)=(x,1)$ for all correct $v$.
  \item \textbf{Graded Agreement}: If $g_w=1$ for some correct $w$, then $y_v=y_w$ for all correct $v$.
\end{itemize}
\end{definition}
\begin{restatable}[implicit in \cite{berman89consensus}, see also \cite{feldman97probabilisitic}]{lemma}{graded}\label{lem:graded}
Graded agreement can be solved in $2$ rounds, with each node sending an $O(\log |\mathcal{V}|)$-sized message to each other node in each round.
\end{restatable}

Finally, we will make use of variants of the central building block of the Phase King algorithm from~\cite{berman89consensus}, which can be generalized to solve the following weaker form of multi-valued consensus.
\begin{definition}[King Consensus]
In the \emph{King Consensus} problem, each node $v$ has input $(x_v,\ell_v)\in \mathcal{V}\times (V\,\dot{\cup}\,\{\bot\})$ and computes output $y_v\in \mathcal{V}\,\dot{\cup}\,\{\bot\}$ with the following guarantees:
\begin{itemize}
 \item \textbf{Validity}: If there is $x\in \mathcal{V}$ so that $x_v=x$ for all correct $v$, then $y_v\in \{x,\bot\}$ for all correct $v$.
 \item \textbf{Default}: If $\ell_v=\bot$ for all correct $v$, then $y_v=\bot$ for all correct $v$.
 \item \textbf{King Agreement}: If there is correct $\ell \in V$ so that $\ell_v=\ell$ for all correct $v$, then $y_v=y_{\ell}\neq \bot$ for all correct $v$.
\end{itemize}
\end{definition}
King consensus relaxes classic (multi-valued) consensus in two ways.
If all nodes' inputs agree on the same correct leader $\ell$, intuitively it behaves like consensus:
all outputs agree, and if $x_v=x$ for all nodes $v$, then this output is going to be $x$.
However, we relax the validity condition in that if $x_v=x$ for all nodes $v$, then it is sufficient to output $x$ or $\bot$.
This provides leeway to save on communication when nodes think that there is no need to execute consensus at all, represented by $\ell_v=\bot$.
The second relaxation is that agreement is only guaranteed if all nodes' inputs agree on some correct leader $\ell=\ell_v$, allowing for an $O(1)$-round solution that still meets the requirements if the nodes do not agree on the filtered clock values.
\begin{restatable}[implicit in \cite{berman89consensus}]{lemma}{king}\label{lem:king}
King Consensus can be solved in $3$ rounds, where each node sends an $O(\log |\mathcal{V}|)$-sized message to each other node in each round.
\end{restatable}
For the sake of completeness, explicit proofs of \Cref{lem:filtering,lem:graded,lem:king} are provided in~\Cref{app:implicit}. The reader is invited to use them as a warm-up before the more involved procedures, but advised to do so after familiarizing themselves with the conventions stated in~\Cref{sec:conventions}.

%% file: overview.tex
\section{Technical Overview}\label{sec:overview}
In this section, we present the high-level ideas underlying our results.
The focus is on the main obstacles and solutions.
We follow the same order as in the subsequent sections, so this outline sketches the progression of the technical exposition in the remainder of the paper.
\paragraph{The big picture}
Essentially, solving counting boils down to solving consensus on the current values of the output counters or clocks:
if they all agree, they should not be changed except for being incremented deterministically by one modulo $C$ each round;
otherwise, \emph{any} agreed-on value is fine.
This corresponds to the validity and agreement properties of consensus, but runs into the issue that there is no (guaranteed) agreement as to \emph{when} to run consensus.
Comparatively simple solutions can be constructed by initializing a new instance to consensus \emph{every} round.
However, this is inherently inefficient in terms of communication, as it incurs the communication cost of a complete run of a consensus protocol per round.

Other approaches rely on randomization to get ``lucky'' and achieve sufficient coordination, as e.g.\ done in~\cite{lenzen13socs}. However, any such approach has the structural properties subjecting it to the same trade-off between the probability of stabilization and the time to achieve it as the one between round complexity and probability of success for consensus~\cite{chor89simple}.

The third type of method that has been used in the literature is to partition the node set into $V_0$ and $V_1$, solve the task recursively on each of these sets, and then use the obtained counters to coordinate when consensus instances are initialized, cf.~\cite{lenzen19counting}. Using a balanced partition limits the recursion depth, number of instances each node participates in, and incurred overheads to $O(\log n)$. We follow this strategy as the only one being compatible with our aims. While we have to swap out every component used in~\cite{lenzen19counting} to achieve our main result, it is instructive to start from their solution and evolve it step by step.

Conveniently, it is trivial to solve counting when $n=1$, so the challenge is to (efficiently) solve counting on $n=2^k$ nodes given a working algorithm for $n/2=2^{k-1}$ nodes. If in neither of the two sets $V_b$ a third of the nodes would be faulty, both sets would stabilize to produce counting outputs. The nodes of $V_b$ would send their current output value to all nodes, which would adopt the majority value as (perceived) ``clock'' $C_b$ for each $b$. We could use these clocks to initialize a consensus instance that terminates after $R$ rounds every $2R+1$ rounds, where the time to run the instance as well as the input given to it are given by the value taken by $C_b$. We could then set the current output clock value to the output of any terminating instance plus $R\bmod C$; otherwise, the output is simply incremented by $1\bmod C$. While the clocks and hence the instances might be arbitrarily aligned relative to each other, one of them is going to terminate while no ``companion'' instance controlled by the other clock is running. From then on, all clock values agree, so that future consensus instances do not affect the count, as they always end up outputting the already agreed-on input, which is then incremented by exactly $R\bmod C$.

The obstacle this strawman approach runs into is that there is no guarantee that $C_b$ counts correctly for both $b\in\{0,1\}$, or even that all correct nodes perceive the same majority value sent by nodes in $V_b$. Abstractly speaking, we have two supernodes $V_0$ and $V_1$ each supplying a clock, but one of them may be Byzantine, sending arbitrary values. In essence, the authors of~\cite{lenzen19counting} address this by preventing equivocation via applying crusader broadcast (a.k.a.\ crusader agreement, see~\cite{dolev82again}). This means to ensure that there is a unique (clock) value $c\in [C]$ such that each correct party either accepts $c$ as the value sent by $V_b$ or outputs $\bot$ (indicating that $V_b$ is Byzantine or has not stabilized yet). However, ``equivocation'' here needs to be interpreted in a broader sense: if what $V_b$ sends is not counting modulo $C$, this is also showing that something is still amiss, and recipients will output $\bot$ whenever the count was off within the last $X$ rounds. This is captured by the $X$-round crusader agreement property of the clock filtering task.

This does not yet fully overcome the above obstacle, as now some correct nodes may participate in a consensus instance triggered by $C_b$ and others do not. While it is possible to require that any already established agreement is not broken this way, the ``bad'' instance controlled by $C_b$ might be initialized before one controlled by $C_{1-b}$ terminates. Hence, some nodes might change their output clocks as a result, breaking the agreement that has just been reached with the help of the instance controlled by $C_{1-b}$. To address this, the final tweak is to have the two clocks $C_b$ initialize instances at slightly different rates, e.g.\ every $2R$ and $3R$ rounds, respectively. Then, the above scenario implies that the next instance controlled by $C_{1-b}$ runs without the same kind of interference and achieves stabilization of the output clock values.

In summary, the strategy from~\cite{lenzen19counting} works as follows.
\begin{enumerate}
  \item Partition $[n]=V_0\dot{\cup}V_1$.
  \item Recursively let $V_b$ count modulo $C\in\Theta(n)$.
  \item Use these outputs as inputs to $(C,\Theta(n),V_b)$-filtering.
  \item Use the outputs of filtering to initialize consensus instances (where even if some nodes do not participate, validity holds).
  \item Use any outputs of such a consensus instance to adjust the output clock; otherwise increment by one modulo $C$. 
\end{enumerate}

We stress that naively replacing the consensus routine with an early-stopping one is insufficient for our purposes:
we do not know at which frequency to initialize fresh instances to (i) achieve stabilization time $O(f+1)$ yet (ii) avoid that instances remain well-separated.

\paragraph{Reducing stabilization time}

In light of the above, to make the stabilization time of a scheme like this $O(f+1)$, we need to use (i) $(C,O(1),V_b)$-filtering and (ii) a consensus subroutine that runs for $O(1)$ rounds only. Alas, achieving consensus deterministically takes at least $t+1$ rounds, so we need to relax the task somehow. Hence, the first crucial adaptation is to replace consensus by king consensus with varying ``kings'' or \emph{leaders} $\ell$, where $F_b$ tells us for each instance which leader to use. Note that king consensus relaxes the agreement property of consensus to hold only when the leader is correct. Fortunately, we only need the (king) agreement property \emph{once,} to establish agreement on the count; afterwards, validity ensures that the output clocks are never modified again by a call to king consensus and hence count correctly, regardless of whether future leaders are correct. Each call to king consensus requires only $O(1)$ rounds, achieving (ii). At the same time, the output clocks stabilize at the latest when a correct leader is reached, i.e., within $O(f+1)$ rounds---counting from the round in which the clock filtering subroutine stabilized.

As already indicated, to achieve fast stabilization for clock filtering, we need to choose $X=O(1)$, i.e., ``forgive'' an incorrect count within $O(1)$ rounds. This allows the adversary to make the value of $F_b$ ``jump'' throughout the course of the stabilization process if $V_b$ contains at least $|V_b|/3$ faulty nodes. To address this, we choose $X$ as a large enough constant such that the additional ``collisions'' that the adversary can induce between instances controlled by $F_b$ and $F_{b-1}$ are limited to a small fraction, say $1/4$, of the instances. While this may prevent \emph{some} correct leaders from enforcing stabilization, among $2f$ instances there will be at least $f$ correct leaders, so that at least $2f-f-2f/4>0$ succeed in stabilizing the output clock values.

Note that this recursion scheme succeeds based on the \emph{first} clock $C_b$ that stabilizes. Accordingly, we can determine a stabilization time bound based on the recurrence $S(f,n)\le O(f)+S(\lfloor f/2\rfloor,\lceil n/2\rceil)$, leading to $S(f,n)\le O\left(\sum_{i=0}^{\lceil \log n\rceil} \max\{f/2^i,1\}\right)=O(f+\log n)$. One could remove the additive $O(\log n)$ by alternating between partitioning evenly and splitting off a single node in the recursion, as then the summation would end once $f/2^i<1$. We postpone this, as such an imbalanced split poses its own challenges when also trying to control bit complexity.

\paragraph{Reducing amortized bit complexity}
In order to achieve an amortized bit complexity of $\tilde{O}(n(f+1))$ per round in our framework, we need to reduce the bit complexities of both the filtering procedure and king consensus. The first observation is that, so long as we have $|V_0|\approx |V_1|\approx n/2$, it is sufficient to do so only when $V_b$ contains fewer than $|V_b|/3$ faults for \emph{each} $b\in \{0,1\}$, as otherwise $\tilde{O}(n(f+1))=\tilde{O}(n^2)$. Similarly, we do not need to worry about the communication cost during the stabilization phase, as the construction uses small messages. During $\tilde{O}(f+1)$ rounds, correct nodes send $\tilde{O}(n^2(f+1))$ bits, contributing $\tilde{O}(n(f+1))$ to the average over $n$ rounds, which is our target.

This is excellent news, as it allows us to work with the assumption that the recursively constructed counters $C_b$ both already count and all correct nodes agree on their values. Thus, king consensus instances are initialized with all correct nodes in agreement on when they start and which party is the leader. In particular, there are only $O(f)$ instances with faulty leaders within $n$ rounds, i.e., such instances also contribute $\tilde{O}((f+1)n)$ bits amortized over $n$ rounds.

Accordingly, we only need to control the communication cost of king consensus instances with agreed-on correct leaders and those without a leader.\footnote{As validity must hold even if nodes do not agree on who the leader is or even whether an instance is running, we formalize this as an instance being initiated in \emph{every} round, but nodes using input $\bot$ to indicate that they have no leader for this instance. This permits correct nodes to ``help'' in achieving the properties if necessary even when they believe that no instance should be started.} To do so, let us revisit the implementation of king consensus derived from the work introducing the Phase King algorithm~\cite{berman89consensus} for the special case that all nodes agree on a leader $\ell$, i.e., $\ell_v=\ell\neq \bot$ for all $v\in [n]$.
\begin{enumerate}
  \item Perform graded agreement on the input values $x_v$, resulting in outputs $(y_v,g_v)$.
  \item The leader $\ell$ sends $y_{\ell}$ to all nodes.
  \item If $g_v=1$, $v$ outputs $y_v$. Otherwise, $v$ outputs $y_{\ell}$.
\end{enumerate}
The reader can readily verify that the guarantees of graded agreement ensure the validity and king agreement properties of king consensus.\footnote{Outputs of $\bot$ and the default property handle the case that some or all nodes have input $\ell_v=\bot$, which we ruled out for the special case we consider here.} The leader's transmission is uncritical from the perspective of communication cost, but graded agreement is costly, requiring $\Omega(n^2)$ bits.

To reduce the bit complexity of king consensus in this special case, we first have the king determine whether there is actual \emph{need} to establish agreement, by having each node send their input to the king. If the system has stabilized, the king will receive the same input from $n-t$ nodes and conclude that an expensive graded consensus is not necessary.

However, the king might reach this conclusion also prior to stabilization, with up to $t$ correct nodes having a differing minority input value. To these nodes, the king points out that their input value differs from the majority. While the king cannot be trusted blindly, this is justification for these nodes to ask all others for their input and for them to respond: we already pointed out that we can bear a cost of $\tilde{O}(n^2)$ bits if the king is faulty, and faulty nodes querying correct nodes for their input can cause no more than $\tilde{O}(nf)$ bits of communication. When a querying node receives at least $n-2t\ge t+1$ values different from its own, it can safely adopt this value, as this is proof that its own input is not shared by all correct nodes. Thus, if the king concludes that no graded consensus is necessary, all correct nodes will adopt the majority value.

On the other hand, if the king receives fewer than $n-t$ times the same value, this is proof that stabilization has not yet occured. In this case it triggers graded agreement. Note that this step empowers a faulty king to invoke graded agreement with only a subset of the correct nodes taking part. We resolve this issue by defining a weaker version of graded agreement. Here, formally all nodes take part, but the graded agreement property only holds if all correct nodes have an additional input bit $s_v$ equal to $1$. A correct king uses the latter to create the behavior of a ``regular'' graded agreement when needed; the validity property remains unchanged, preventing a faulty king from disrupting existing agreement. This relaxation allows an implementation in which node $v$ stays silent, i.e., sends no messages, if $s_v=0$, achieving the desired reduction in communication cost if all correct nodes share the same input to the king consensus instance.

To achieve filtering at amortized cost of $\tilde{O}(n(f+1))$ bits, recall that we may assume that $f<|V_b|/3$ and that the counting instance on $V_b$ has already stabilized. Hence, all but $f$ nodes in $|V_b|$ are correct and agree on $C_b$. Node $v$ maintains a ``guess'' of the current value of $C_b$ and memorizes what it believes this value to be at all other nodes. If the guess has been counting for $X$ rounds and no inconsistency was proven, it will output it, otherwise it will output $\bot$. The algorithm now carefully implements several checks and queries to other nodes' current guesses or $C_b$ values (for those in $V_b$) to ensure that (i) if the guess of $v$ is incorrect, it will correct it within $O(f+1)$ rounds and (ii) after stabilization, only $\tilde{O}(n(f+1))$ bits are sent on average.

The principal ideas to achieve this are the following.
\begin{enumerate}
  \item Node $v$ queries nodes $w\in V_b$ for the $C_b$ value on a round-robin basis, one per round ($\tilde{O}(n)$ bits). If $w$ responds with a value different from $v$'s current guess, $v$ queries all nodes, causing an amortized cost of $\tilde{O}(nf)$ bits.
  \item If $v$ believes that a majority of nodes in $V_b$ has a different value than its current guess, it adjusts its guess, sends its new guess to all nodes, and queries all nodes (for verification purposes). This happens at most twice after stabilization of $C_b$, as after a query to all nodes $v$ updated all stale information. Hence, the amortized cost is $\tilde{O}(n)$ bits. Together with the first rule, this guarantees that all nodes adopt the correct guess within $f+O(1)$ rounds.
  \item If a node has a guess that is different from the locally memorized guess of another node, it queries that node and updates its memory based on the response. Once all correct nodes adopted the correct guess and communicated about this, this costs only $\tilde{O}(nf)$ bits per round, since there are no queries between correct nodes anymore.
\end{enumerate}
This last rule enforces that correct nodes with different guesses notice the discrepancy. If a node believes that more than $n/3$, i.e., more than $t$, guesses deviate from its own, it can conclude that not all correct nodes agree on $C_b$ and will output $\bot$ at least for the next $X$ rounds. This guarantees that there is only one non-$\bot$ value output by correct nodes within $X$ rounds, i.e., the $X$-round crusader agreement property is satisfied. On the other hand, if $V_b$ contains fewer than $|V_b|/3$ faulty parties, all correct parties adopt the majority value in $V_b$ and will synchronize their views of each other's guesses. Afterwards, the will not perceive any inconsistencies and start outputting the majority value within an additional $X\in O(1)$ rounds, showing the validity property of clock filtering.


\paragraph{Achieving asymptotically optimal stabilization time}
As mentioned earlier, removing the additive $O(\log n)$ from the stabilization time is challenging when seeking to simultaneously keep the amortized bit complexity small. Our approach is to use a recursion template where one branch has only a single node $\ell$, which can easily generate and distribute a count (if correct), while the other branch generates its clock $C$ by recursing on the \emph{entire} node set. By alternating between the two templates 
we ensure stabilization in $O(1)$ rounds once the current node set contains no faults, yet cut the number of faults in the branch with fewest faults in half with every other level of recursion, maintaining the geometric decay in stabilization time on the ``fastest'' branch.

The primary obstacle to keeping a small communication footprint with this approach is that if $\ell$ is faulty, we cannot prevent it from manipulating the clock it generates to make itself leader every $X+1$ rounds. However, we need that $X\in O(1)$ to ensure fast stabilization if $f=0$. Hence, our previous strategy of amortizing the cost incurred by instances with faulty leaders is not going to work here. In fact, we can give up on filtering altogether and accept that $\ell$ is the leader in every instance, letting it choose when to initiate king consensus on its own. If correct, $\ell$ will choose a time that prevents an overlap with an instance controlled by $C$, and we let nodes ignore $\ell$ if it attempts to initiate an instance that would collide with one controlled by $C$.

To avoid paying too much for the king consensus instances lead by $\ell$, we relax the king agreement property to what we term \emph{weak king agreement:} only if there are no faults whatsoever, $\ell$ needs to be able to get correct nodes to adopt the same value if their inputs differ. This relaxation to the \emph{weak king consensus} task is irrelevant for our stabilization mechanism, as (weak) king agreement is only needed to stabilize the output clocks, and for this recursive pattern, the single-node branch of $\ell$ needs to ensure stabilization only if there are no faults.

On the other hand, this provides us with just enough leeway to ensure that weak king consensus can be implemented at bit complexity $\tilde{O}(n(f+1))$ even if $\ell$ is faulty. Our strategy here is to not rely on $\ell$ for ``justifying'' communication by correct nodes. Instead, we employ the following strategy, where we describe the special case that no node ignores the leader.
\begin{enumerate}
  \item Nodes compare their current clock value to their neighbors in a constant-degree expander. If stabilization has been achieved, i.e., all correct nodes have the same clock value, no more than $O(f)$ nodes, those with a faulty neighbor in the expander, can observe a discrepancy. On the other hand, if $f=0$ and $k>0$ nodes have a clock value different from the majority value, at least $\varepsilon k$ nodes, where the constant $\varepsilon$ depends on the expander, observe a discrepancy.
  \item All nodes observing a discrepancy alert all nodes of this. If $v$ receives $k_v$ such alert messages, it queries each node $w\in \{v,\ldots,v+\lceil 2k_v/\varepsilon\rceil\bmod n\}$ for its clock value; these nodes respond with their clock values. If the system has already stabilized, all of this incurs only $\tilde{O}(n f)$ bits of communication: by the previous point, then $k\in O(f)$ at each node $v$, so there are $O(nf)$ notification, query, and response messages in total.
  \item All nodes for which half of their queries resulted in responses with clock values different from their own now query \emph{all} nodes for their clock values, which respond. Again, by the choice of the communication pattern in the previous step, if all correct nodes agree, this does not cause more than $O(f)$ nodes to perform such queries, resulting in $O(nf)$ messages in total. On the other hand, if $f=0$, each node that does not have the majority clock value will be among the querying nodes in this step and learn that at least half of the nodes (claim to) have a different clock value. In particular, more than $t$ nodes have a different clock value, proving to them that it is safe, i.e., will not violate validity, to adopt a clock value proposed by $\ell$.
  \item This is, finally, the point where $\ell$ steps in and proposes the majority value it perceives to all nodes. If $f=0$, this \emph{is} the actual majority value. All nodes that learned in the previous step that it is safe to change value adopt the proposal by $\ell$, which in case of $f=0$ establishes agreement, i.e., weak king agreement is satisfied. All other nodes simply output their input value, guaranteeing validity.
\end{enumerate}

%% file: conventions.tex
\section{Conventions for Pseudocode}\label{sec:conventions}
We will employ subroutines that are not self-stabilizing, but run only for a fixed number of rounds $R\in \NN$. We refer to such a subroutine as an \emph{$R$-round algorithm.} Note that it is trivial to obtain an $R$-round algorithm from any algorithm of round complexity at most $R$, by instructing nodes to wait until round $R$ before terminating and producing their output. There is also a fixed bound on the number of subroutine calls that are initialized per round. Each such instance has its own dedicated memory for each of its rounds, where we tacitly assume that there is a known bound on the amount of memory the subroutine requires. This approach can be seen as a implementing ``self-stabilizing pipeline'' for $R$-round algorithms: within $R$ rounds, all inconsistent state from a period of transient faults is erased. Note, however, that we cannot enforce a globally consistent initialization without a common round counter, which is the key challenge of this work.

Moreover, we will run self-stabilizing algorithms as subroutines on a subset of the nodes as part of our recursive solutions. Again, there is a fixed number of such routines. We emphasize this, because it requires care when seeking to reduce the amortized bit complexity to $o(n^2)$: this requires that in many rounds, most pairs of nodes do not exchange any messages. To make this possible while freely combining subroutines, the main routine gathers all messages to be sent to a given destination and concatenates them as a list of pairs consisting of a bit string identifying a subroutine and the corresponding message. If no message is sent by a given subroutine in a given round, we can safely omit the respective pair from the list. Whenever a subroutine sends a message, we will assume that it has size $\Omega(\log n)$ for analysis purposes (even if it is smaller). Since in our recursive algorithms each node participates in $O(\log n)$ different subroutines, this approach bears no asymptotic cost on communication bounds, but greatly simplifies the description of our routines.

With this convention, we can leave the association of subroutine messages to the compiler. In fact, we will allow a constant number of send instructions to be executed by the same (sub)routine in a single round, tacitly assuming that they are treated in the same way. Moreover, we will occasionally send ``special'' messages that have a unique name, e.g.\ ``req,'' but no content. Again, we leave it to the compiler to choose an encoding such that these cannot be misinterpreted as, e.g., a clock value being sent. For convenience, we also assume that the compiler will choose an encoding that is not too wasteful, i.e., we can assume that if $k$ different messages can possibly be sent by a (sub)routine, the encoding usese $O(\log k)$ bits. Faulty nodes may of course deviate from these rules. However, the compiler will instruct nodes to drop any received bit string that is not encoding a valid message. Note that it may still happen that due to an inconsistent (global) state or more than $t$ faults,\footnote{While we assume that $f\le t$ during our analysis, note that subroutine calls on $n'<n$ nodes may have to contend with $f'>n'/3$ faults.} code might produce out-of-spec results or not be able to execute at all. To address this, we make the following assumptions:
\begin{itemize}
  \item The compiler will insert checks for all state variables to make sure that the stored value encodes a valid value. If this is not the case, it resets the variable to an arbitrary in-spec default before executing the code.
  \item $R$-round algorithms that are used as subroutines are given by a state machine. This state machine takes the round number, input (for round $1$) or state at the end of the previous round (for rounds $2,\ldots,R$) to compute the messages it sends in that round, and with the additional input given by messages received from other nodes computes the state at the end of the round (for rounds $1,\ldots,R-1$) or output (for round $R$). For all our non-stabilizing subroutines, it will be obvious from the description that they can be realized this way with bounded memory. The number of required rounds $R$ will also be readily visible.
  \item Pseudocode for self-stabilizing (sub)routines executes successully for any valid state and received messages. Again, this will be obvious from the description of our self-stabilizing (sub)routines. A more subtle point in this regard is that the code must also execute \emph{in a single round,} as we cannot rely on the node knowing ``which round it is,'' i.e., any notion of round progression must be explicitly captured by local variables that can be affected by transient faults. To address this, we require that once a variable has been changed by an instruction that is conditioned on message reception, it cannot be used in any subsequent send instruction. Equivalently, when executing the code for a self-stabilizing routine in round $r$, the content of each message sent is a function of the local variables' states $X_{v,r-1}$ at the end of round $r-1$.
  \item Bounds on message size are maintained even if $f>t$. This is the only property we use in subroutine calls on node sets with too many faults, and it will be clear from our descriptions of algorithms as well.
\end{itemize}

Moreover, we will frequently use statements of the form ``if received $k$ times value $x$.'' In such statements we tacitly assume that the respective messages are sent by distinct senders. In particular, if a faulty node sends more than one message in a given round, the compiler ensures that a correct receiver will process only the first one that arrives. Moreover, for notational convenience, nodes may send messages ``to themselves;'' in particular, a statement like ``send $x$ to all nodes'' executed by $v$ will trigger a subsequent ``if received $x$ from $w$'' condition at $v$ with $w=v$ at node $v$.

%% file: counting.tex
\section{Early-Stabilizing Counting}\label{sec:counting}

In this section, we develop the main approach for achieving stabilization time of $O(f+1)$, leaving the issue of small amortized bit complexity for later. Observe that for a single node, $C$-counting is trivial to solve with $S=1$: the solitary node increments a counter modulo $C$ in each round and outputs its value. We will now solve $C$-counting recursively. To this end, we establish a general template for recursive $C$-counting, whose pseudocode is given in~\Cref{alg:recursion}.

As outlined in \Cref{sec:overview}, we partition the node set into $V_0$ and $V_1$, execute a counting algorithm on each of them, and pass the outputs through filtering subroutines. We use \emph{their} output to initiate instances of king consensus, cycling through leaders, but at slightly different frequencies determined by constants $k_b$, where $b\in \{0,1\}$.

\begin{algorithm2e}[b!]
	\caption{A recursion template for $C$-counting on node set $V$. We provide the code $v\in V$ executes in each round. The positive integers $X,R,k_0,k_1\in O(1)$ will be fixed later. The template is parametrized a partition $V=V_0\,\dot{\cup}\,V_1$, solutions to counting on these sets, to filtering with them as clock sets, and to $R$-round king consensus on node set $V$.\label{alg:recursion}}
	\KwVars{for $b\in\{0,1\}$, outputs $C^b_v$ and $F^b_v$ of counting on $V_b$ and filtering with clock set $V_b$}
	\KwOut{$C_v\in [C]$}
	\ForEach{$b\in \{0,1\}$}{
		\leIf{$F^b_v\bmod k_b n=k_bw$}{$\ell_v^b:=w$}{$\ell_v^b:=\bot$}
		initialize king consensus on universe $[C]\cup \{\bot\}$ with input $(C_v+R\bmod C,\ell_v^b)$\;
		\For{$r\in \{R,\ldots,1\}$}{
			execute the code of round $r$ of king consensus on the state stored for round $r$ and $b$\;
			\If{$r=R$ and the output of the instance is $c\in [C]$}{
				$C_v:=c$\;
			}
			\Else{
				store the new state in the memory block allocated for round $r+1$ and $b$\;
			}
		}
		\If{$v\in V_b$}{
			execute the code of an instance of $(k_b n,X,V_b)$-filtering with input $C^b_v$ and output $F^b_v$\;
			execute the code of an instance of $(k_b n)$-counting on node set $V_b$ with output $C^b_v$\;
		}
		\Else{
			execute the code of an instance of $(k_b n,X,V_b)$-filtering with output $F^b_v$\;
		}
	}
	$C_v:=C_v+1\bmod C$\;
\end{algorithm2e}

In the following, for $b\in \{0,1\}$ denote by $S^b_C$ the stabilization time of the instance of $(k_b n)$-counting on node set $V_b$. Similarly, $S^b_F$ is the stabilization time of the instance of $(k_b n,X,V_b)$-filtering. First, let us note that if there are not too many faults within $V_b$, then the variables $F^b_v$ stabilize and begin to count.
\begin{lemma}\label{lem:filter_good}
Suppose that for $b\in\{0,1\}$, it holds that $f<|V_b|/3$. Then the variables $F^b_v$ $(k_b n)$-count from round $r\ge S^b_C+S^b_F$.
\end{lemma}
\begin{proof}
Within $S^b_C$ rounds, the variables $C^b_v$ begin to count. By validity of clock filtering, the variables $F^b_v$ thus count from round $S^b_C+S^b_F$.
\end{proof}
On the other hand, even if there are too many faults in $V_b$, the filtering subroutine ensures that there is a ``virtual clock'' $\hat{F}^b_r\in [k_b n]$ that counts for during each $X$ consecutive rounds, such that correct node $v$ either outputs $F^b_{v,r}=\hat{F}^b_r$ or $F^b_{v,r}=\bot$ in round $r$.
\begin{lemma}\label{lem:filter_bad}
For $b\in \{0,1\}$ and rounds $r\ge S^b_F$, there is a \emph{virtual clock} $\hat{F}^b_r\in [k_b n]$ such that (i) for all correct nodes $v$, it holds that $F^b_{v,r}\in \{\hat{F}^b_r,\bot\}$ and (ii) unless $r-r_0 = 0 \bmod X$, $\hat{F}^b_r=\hat{F}^b_{r-1}+1\bmod k_b n$.
\end{lemma}
\begin{proof}
To define $\hat{F}^b$, consider in each interval $I_i:=[r_0+iX,r_0+i(X+1)-1]$ the minimal round $r\in I_i$ such that $F^b_{v,r}\neq \bot$ for some correct $v$ and set $\hat{F}^b_{r'}:=F^b_{v,r}+r-r'$ for all $r'\in I_i$ (if no such $r$ exists, use an arbitrary value instead of $F^b_{v,r}$). By the $X$-round crusader agreement property of clock filtering, this results in a well-defined $\hat{F}^b$ meeting the requirements.
\end{proof}

The main pillar of the proof of stabilization of \Cref{alg:recursion} is showing that eventually, there will be a ``good'' instance of king consensus, in the sense that all nodes agree on a correct leader and there is no ``competing'' instance that is initialized during its execution.
\begin{definition}[Unimpeded Instances]\label{def:unimpeded}
For $b\in \{0,1\}$, we say that the $b$-th instance of king consensus initialized in round $r$ \emph{runs with leader $\ell\in V$} if (i) the variables $F^b_v$ count from round $r$ and (ii) each correct node $v$ inputs $(x_v,\ell)$ for some $x_v\in [C]$. Furthermore, such an instance is \emph{unimpeded} if in addition (i) $\ell$ is correct and (ii) correct nodes $v$ initialize all other instances of king consensus in rounds $r'\in \{r,r+1,\ldots,r+R\}$ with $\ell_v= \bot$.
\end{definition}
Before proving that such an instance will occur in a timely fashion, let us formalize that an unimpeded instance translates to stabilization.
\begin{lemma}\label{lem:stabilization}
Assume that in round $r\in \NN$ an unimpeded instance of king consensus is initialized. Then the variables $C_v$ count from round $r+R$.
\end{lemma}
\begin{proof}
By the king agreement property of $R$-round king consensus, there is some $y\in [C]$ such that all correct nodes $v$ output $y_v=y$ from the unimpeded instance. Observe that, granted that no king consensus instance outputs a value different from $y+r'-r\bmod C$ or $\bot$ in a round $r'\ge r$ at some correct $v$, correct nodes will set $C_{v,r}:=y$ and, by induction, the variables $C_v$ will count from round $r$.

To show this, recall that the definition of an unimpeded instance requires that correct nodes $v$ have $\ell_v=\bot$ in all other king consensus instances initialized in rounds $r'\in \{r,r+1,\ldots,r+R\}$, i.e., from the initiliazation of the unimpeded instance until it outputs in round $r+R$. By the default property of king consensus, this entails that all other instances initiated in these rounds output $\bot$ at correct nodes.

Accordingly, it remains to consider instances that are initiated in rounds $r'>r$. Assume for contradiction that there is a minimal round $r'>r$ such that an instance initiated in round $r'$ outputs $y_v\notin \{y+r'-r+R\bmod C,\bot\}$ at some correct $v$ (which it does in round $r'+R$). As this instance is initialized after round $r$, but is the first instance that outputs $y_v\notin \{y+r'-r+R\bmod C,\bot\}$, correct nodes used inputs $x_v\in \{y+r'-r+R\bmod C,\bot\}$ to the instance. By the validity property of king consenus, they must output $y_v\in \{y+r'-r+R\bmod C,\bot\}$, reaching a contradiction and concluding the proof.
\end{proof}
In order to show that suitable parameter choices result in quick emergence of an unimpeded instance, we first state a straightforward, but technical helper lemma for said choices.
\begin{lemma}\label{lem:constants}
Suppose that a constant $0<\varepsilon<1$ and natural $R\in O(1)$ are given.
Then there are natural $R<k_0,k_1\in O(1)$ with the following property.
For any choices of $o\in \ZZ$, $b\in \{0,1\}$, and $L_{\max}\in \ZZ_{>0}$, it holds that
\begin{equation*}
\left|\left\{L\in \NN, L\le L_{\max} \mid [k_b\cdot L,k_b\cdot L+R]\cap \{k_{1-b}\cdot L'+o\}\,|\,L'\in \ZZ\}\neq \emptyset\right\}\right|\le \varepsilon L_{\max}+1.
\end{equation*}
\end{lemma}
\begin{proof}
Let $k_0:=\lceil (R+1)/\varepsilon \rceil$ and $k_1:=k_0+R+1$.
Let $L\in \NN$ be minimal with the property that
\begin{equation*}
[k_b\cdot L,k_b\cdot L+R]\cap \{k_{1-b}\cdot L'+o\,|\,L'\in \ZZ\}\neq \emptyset.
\end{equation*}
By the choices of $k_b$, $b\in \{0,1\}$, it follows that for $L+1,L+2,\ldots,L+\lceil 1/\varepsilon\rceil-1$ the intersection is empty.
Repeating this argument inductively, the claim follows.
\end{proof}
We are now ready to resolve the choices of constants that guarantee an unimpeded instance $O(f+1)$ rounds after one of the recursive counting instances and both filtering subroutines have stabilized.
\begin{lemma}\label{lem:unimpeded}
If $R\in O(1)$, there are constants $X,k_0,k_1$ with the following property. If for $b\in \{0,1\}$ it holds that $f<|V_b|/3$, there is an unimpeded instance within $\max\{S^b_C+S^b_F,S^{1-b}_F\}+O(f+1)$ rounds.
\end{lemma}
\begin{proof}
By \Cref{lem:filter_good}, the variables $F^b_v$ count from round $S^b_C+S^b_F$. By \Cref{lem:filter_bad}, there is a virtual clock $\hat{F}^{1-b}$ such that for each $i\in \NN$ and $r\in I_i:=[S^{1-b}_F+iX+1,S^{1-b}_F+(i+1)X-1]$, it holds that $\hat{F}^{1-b}_r=\hat{F}^{1-b}_{r-1}+1\bmod k_b n$ and for each correct node $v$, we have that $F^{1-b}_{r,v}\in \{\hat{F}^{1-b}_r,\bot\}$ for all $r\ge S^{1-b}_F$. We apply \Cref{lem:constants} as follows:
\begin{itemize}
  \item Fix $k_0,k_1$ as given by the lemma for $\varepsilon=1/4$ and the given value of $R$.
  \item Fix $X=5\max\{k_0,k_1\}+R$ and $L_{\max}:=\lfloor (X-R)/k_b\rfloor-1>4$.
  \item Fix $i\in \NN$ such that $S^{1-b}_F+iX\ge S^b_C+S^b_F$. Let $r_b\in I_i$ be minimal such that there is $\ell_b\in V$ with $F^b_{v,r_b}=k_b\ell_b \bmod k_b n$ for some correct node $v$.
  \item Let $r_{1-b}\in I_i$ be minimal such that there is $\ell_{1-b}\in V$ with $\hat{F}^{1-b}_{r_{1-b}}=k_{1-b}\ell_{1-b} \bmod k_{1-b} n$ for some correct node $v$. Choose $o:=r_{1-b}-r_b$.
\end{itemize}
Now consider a leader $\ell \in \{\ell_b,\ell_b+1,\ldots,\ell_b+L_{\max}\}$, where for convenience we identify values $\ell>n$ with the node $v$ such that $v=\ell\bmod n$. Observe that $r_b+(\ell-\ell_b)k_b\ge r_b$ and
\begin{equation*}
r_b+(\ell-\ell_b)k_b+R\le r_b+L_{\max}k_b+R\le r_b+X-k_b\le S^{1-b}_F+(i+1)X-1,
\end{equation*}
where the last step uses that $r_b$ is one of the first $k_b$ rounds of $I_i$. That is, we have that $J_{\ell}:=\{r_b+(\ell-\ell_b)k_b,r_b+(\ell-\ell_b)k_b+1,\ldots,r_b+(\ell-\ell_b)k_b+R-1\}\subset I_i$.

Recall that the variables $F^b_v$ count from round $S^b_C+S^b_F$, which due to our choice of $i$ is no later than the last round before $I_i$. Hence, the $b$-th instance of king consensus initialized in round $r_b+(\ell-\ell_b)k_b-1$ runs with leader $\ell$, where the substraction of $1$ takes into account that $\ell_v^b$ is determined based on $F^b_{v,r-1}$ in round $r$, i.e., before $F_v^b$ is updated by the filtering subroutine. Because $k_b>R$, each correct node $v$ initializes its $b$-th instance in rounds $r\in J_{\ell}\setminus \{r_b+(\ell-\ell_b)k_b\}$ with $\ell_v=\bot$. We conclude that the instance is unimpeded if and only if (i) $\ell$ is correct and (ii) no correct node $v$ uses input $\ell_v\neq \bot$ on its $(1-b)$-th instance in some round $r\in J_{\ell}$.

Let $i_0\in \mathbb{N}$ be minimal such that (i) $S^{1-b}_F+i_0 X+1\ge S^b_C+S^b_F$, (ii) the variables $F^b_v$ do not ``overflow'' during $I_{i_0}$, i.e., $F^b_{v,r}=F^b_{v,r-1}+1$ for all $r\in I_{i_0}$, and (iii) fewer than half of the leaders $\{\ell_b,\ell_b+1,\ldots,\ell_b+L_{\max}\}$ for the interval $I_i$ are faulty. Since there are at most $f$ faulty nodes and $n>3f$ (i.e., we do not run out of ``fresh'' leaders before exhausting the fault budget),\footnote{Due to condition (ii), this applies as-is only if $n$ is a sufficiently large constant. However, if $n=O(1)$, we may work with clocks that count modulo $m k_b n$ for a sufficiently large constant $m$ instead.} we have that $i_0\in \max\{(S^b_C+S^b_F-S^{1-b}_F)/X,0\}+O(f+1)$. By our choices of $\varepsilon$, $k_0$, $k_1$, and $o$, \Cref{lem:constants} guarantees that out of the remaining at least $L_{\max}/2$ correct leaders $\ell$, all but $L_{\max}/4+1$ satisfy that
\begin{equation*}
[k_b (\ell-\ell_b),k_b (\ell-\ell_b)+R]\cap \{k_{1-b}\cdot L+r_{1-b}-r_b\,|\,L\in \ZZ\}= \emptyset.
\end{equation*}
By our choice of $L_{\max}$, $L_{\max}/2-L_{\max}/4-1>0$, i.e., some correct $\ell$ satisfying the above constraint exists.

We claim that the corresponding $b$-th instance initiliazed in round $r_b+(\ell-\ell_b)k_b$ is unimpeded. Given that $\ell$ is correct, it remains to show that no correct node $v$ initilializes a $(1-b)$-th instance with $\ell_v\neq \bot$ in some round $r\in J_{\ell}$. Because $F^{1-b}_{v,r}\in \{\hat{F}^{1-b}_r,\bot\}$, this is only possible if $\hat{F}^{1-b}_r\bmod k_{1-b}=0$. 

Accordingly, assume for contradiction that $\hat{F}^{1-b}_r\bmod k_{1-b}=0$ for some $r\in J_{\ell}$. Recall that $\hat{F}^{1-b}_{r'}=\hat{F}^{1-b}_r+r'-r\bmod k_{1-b} n$ for all $r,r'\in I_i$ and that $\hat{F}^{1-b}_{r_{1-b}}=0\bmod k_{1-b}$. Therefore, 
\begin{equation*}
0=\hat{F}^{1-b}_r\bmod k_{1-b} = \hat{F}^{1-b}_{r_{1-b}}+r-r_{1-b} \bmod k_{1-b}=r-r_{1-b} \bmod k_{1-b},
\end{equation*}
i.e., $r-r_b\in \{k_{1-b}\cdot L+r_{1-b}-r_b\,|\,L\in \ZZ\}$.
On the other hand, by our choice of $i_0$,
\begin{equation*}
F^b_{v,r}-F^b_{v,r_b} = r-r_b \in [k_b (\ell-\ell_b),k_b (\ell-\ell_b)+R].
\end{equation*}
We reach the contradiction that $r-r_b$ is in the intersection of these two sets, which by \Cref{lem:constants} is empty. We conclude that the instance is unimpeded, as claimed. Since $i_0\in \max\{(S^b_C+S^b_F-S^{1-b}_F)/X,0\}+O(f+1)$, this instance occurs by round $\max\{S^b_C+S^b_F,S^{1-b}_F\}+O(f+1)$.
\end{proof}
Putting the above results together, we arrive at the following theorem.
\begin{theorem}\label{thm:stab}
If $R\in O(1)$, there are constants $k_0$, $k_1$, and $X$ such that \Cref{alg:recursion} solves $C$-counting with stabilization time $\max_{b\in \{0,1\}}\{S^b_C+S^b_F\}+O(f+1)$. If $f<|V_b|/3$ for some $b\in \{0,1\}$, then the stabilization time is $\max\{S^b_C+S^b_F,S^{1-b}_F\}+O(f+1)$.
\end{theorem}
\begin{proof}
We show the second statement first. By \Cref{lem:unimpeded} and the assumptions of the theorem, there is an unimpeded instance of king consensus within $\max\{S^b_C+S^b_F,S^{1-b}_F\}+O(f+1)$ rounds. By \Cref{lem:stabilization}, the output variables then start to count from some round $r\in \max\{S^b_C+S^b_F,S^{1-b}_F\}+O(f+1)$.

Regarding the first statement, recall that $f< n/3$ and $V_0$ and $V_1$ partition $V=\{1,\ldots,n\}$, so there must be some $b\in \{0,1\}$ so that $f<|V_b|/3$. Since $\max\{S^b_C+S^b_F,S^{1-b}_F\}\le \max_{b\in \{0,1\}}\{S^b_C+S^b_F\}$, the claim readily follows from the second statement.
\end{proof}
In the following, we will tacitly assume that the constants in~\Cref{alg:recursion} are chosen in accordance with \Cref{thm:stab} for the $R\in O(1)$ given by the choice of the king consensus algorithm employed. Applying the recursive scheme with the trivial solution on single nodes, we arrive at the following result.
\begin{corollary}\label{cor:stab}
$C$-counting can be solved with stabilization time $O(f+1+\log n)$, where each correct node sends $O(\log C + \log^2 n)$ bits to each other node in each round.
\end{corollary}
\begin{proof}
By \Cref{lem:filtering,lem:king}, there are filtering and king consensus protocols that run in a constant number of rounds and use messages of size $O(\log C)$ and $O(\log |\mathcal{V}|)$, respectively. We recursively apply \Cref{alg:recursion}, where we partition the node set as evenly as possible in each recursion step and the base case of $n=1$ is trivial. In each step of the recursion, either $V_0$ or $V_1$ contains at most half of the faulty nodes. Hence, summing over the at most $\lceil \log n\rceil$ steps along the respective root-leaf path in the recursion tree, the stabilization time is bounded by
\begin{equation*}
\sum_{i=0}^{\lceil\log n\rceil-1}O(2^{-i}f+1)=O(f+1+\log n).
\end{equation*}
The message size bound follows by observing that (i) each node participates in $O(\log n)$ recursive instances, each running a constant number of concurrent instances of filtering and king consensus, and (ii) all of these instances but the top-level king consensus instances, which have message size $O(\log C)$, use messages of size $O(\log n)$.
\end{proof}

We remark that it is straightforward to remove the additive $O(\log n)$ by alternating between splitting off a single node and using a balanced partition (as in \Cref{cor:stab}). However, this will become more challenging when also taking into account amortized bit complexity. We hence defer this to \Cref{sec:opt_stab}.

%% file: communication.tex
\section{Reducing Communication Complexity}\label{sec:communication}
We need to reduce the communication cost of both the filtering and king consensus routines to realize our $\tilde{O}(n(f+1))$ target. We address each of them separately, starting with king consensus.
\subsection{King Consensus with less Communication}

As discussed in \Cref{sec:overview}, it is sufficient to consider the case of a correct king after stabilization only, as all other cases combined do not bust our communication budget. We construct our improved solution to king consensus from relaxed versions of graded agreement. The first relaxes the graded agreement property to hold only if an additional input bit is $1$ at all correct nodes.
\begin{definition}[Weak Graded Agreement]
The \emph{Weak Graded Agreement} problem is specified as follows.
Each node $v$ has input $(x_v,s_v)\in \mathcal{V}\times \{0,1\}$.
Each node $v$ computes an output $(y_v,g_v)\in \mathcal{V}\times \{0,1\}$ with the following guarantees:
\begin{itemize}
  \item \textbf{Validity}: If there is $x\in \mathcal{V}$ so that $x_v=x$ for all correct $v$, then $(y_v,g_v)=(x,1)$ for all correct $v$.
  \item \textbf{Weak Graded Agreement}: If $s_v=1$ for all $v$ and $g_w=1$ for some correct $w$, then $y_v=y_w$ for all correct $v$.
\end{itemize}
\end{definition}
If $s_v=0$, for this relaxed task we can get away with $v$ sending no message.
\begin{lemma}\label{lem:weak_graded}
Weak graded agreement can be solved in $2$ rounds, where each node sends an $O(\log |\mathcal{V}|)$-sized message to each other node in each round. If node $v$ has input $s_v=0$, $v$ sends no messages.
\end{lemma}
\begin{proof}
We claim that the following algorithm achieves the stated guarantees.

\begin{algorithm2e}[H]
	\caption{Weak graded agreement.}
	\KwIn{$(x_v,s_v)\in \mathcal{V}\times \{0,1\}$}
	\KwOut{$(y_v,g_v)\in \mathcal{V}\times \{0,1\}$}
	\lIf{$s_v=1$}{send $x_v$ to all nodes\tcp*[f]{first communication round}}
	\lIf{$s_v=1$ and received at most $t$ messages different from $x_v$}{send $x_v$ to all nodes}
	\lElseIf{$s_v=1$}{send ``NACK'' to all nodes\tcp*[f]{second communication round}}
	\lIf{received at most $t$ messages different from $x_v$}{$(y_v,g_v):=(x_v,1)$}
	\lElseIf(\tcp*[f]{break ties arbitrarily}){received $t+1$ times $x\in \mathcal{V}$}{$(y_v,g_v):=(x,0)$}
	\lElse{$(y_v,g_v):=(x_v,0)$}
\end{algorithm2e}
The running time and bounds on communication can be readily verified from this description.

\textbf{Validity:} If there is $x\in \mathcal{V}$ so that $x_v=x$ for all correct $v$, then each $v$ receives values other than $x=x_v$ from the at most $t$ faulty nodes only. Hence, $(y_v,g_v)=(x,1)$ for all correct $v$.

\textbf{Weak Graded Agreement}: Suppose that $s_v=1$ for all $v$ and $g_w=1$ for some correct $w$. Since $s_v=1$ for all $v$, correct nodes always send messages. Hence, $w$ received at least $n-2t>t$ messages with $x_w$ in the second round from correct nodes. Their senders must have received at least $n-2t>t$ messages with $x_w$ from correct nodes in the first round, or they would have sent ``NACK'' messages. It follows that each correct $v$ with $x_v\neq x_w$ receives more than $t$ values different from $x_v$ in the first round, implying that it will not send $x_v$ again in the second round. Therefore, in the second round each node receives more than $t$ times $x_w$ and at most $t$ times $x$ for any $x\neq x_w$. We conclude that each correct $v$ outputs $y_v=x_w$.
\end{proof}

The second relaxed variant of graded agreement is leader-based. The second output bit, the ``grade,'' is intended for use by the leader only. Accordingly, the validity property does not require grade $1$, and the graded agreement property is relaxed to requiring that if a correct agreed-on leader $\ell$ outputs $g_{\ell}=1$, then all nodes should output the same value as the leader (regardless of grade). An additional property we refer to as \emph{king validity} is what will enable the leader to leverage the reduced cost of weak graded agreement if $s_v=0$: if all nodes agree on the input value \emph{and} a correct leader, the leader is guaranteed to output $g_{\ell}=1$. As graded king agreement implies that $g_{\ell}=1$ guarantees that all nodes already agree, the leader then can decide whether to invoke weak graded agreement with inputs $s_v=1$ or $s_v=0$ based on $g_{\ell}$.
\begin{definition}[Graded King Consensus]
The \emph{Graded King Consensus} problem is specified as follows.
Each node $v$ has inputs $x_v\in \mathcal{V}$ and $\ell_v \in V\,\dot{\cup}\,\{\bot\}$.
Each node $v$ computes an output $(y_v,g_v)\in \mathcal{V}\times \{0,1\}$ with the following guarantees:
\begin{itemize}
  \item \textbf{Validity}: If there is $x\in \mathcal{V}$ so that $x_v=x$ for all correct $v$, then $y_v=x$ for all correct $v$. 
  \item \textbf{King Validity}: If there are $(x,\ell)\in \mathcal{V}\times V$ so that $(x_v,\ell_v)=(x,\ell)$ for all correct $v$ and $\ell$ is correct, then $g_{\ell}=1$.
  \item \textbf{Graded King Agreement}: If $\ell_v=\ell$ for all correct $v$ and some correct $\ell$ and $g_{\ell}=1$, then $y_v=y_\ell$ for all correct $v$.
\end{itemize}
\end{definition}
The requirements of graded king consensus are sufficiently flexible that it suffices for the leader to check everyone's values and advise a change if it looks like they might agree, i.e., it receives $n-t$ times the same value. In this case, it can advise the nodes with different values to change it, which they can verify to be safe (i.e., not violate validity), by confirming that there are at least $t+1\le n-2t$ nodes claiming to have this value. If the everybody agrees on the input values and a correct leader, this causes no communication between correct non-leader nodes.
\begin{lemma}\label{lem:graded_king}
Graded king consensus can be solved within $4$ rounds with message size $O(\log |\mathcal{V}|)$ and the following property. If there is $(x,\ell)\in \mathcal{V}\times(V\cup \{\bot\})$ so that $(x_v,\ell_v)=(x,\ell)$ for all correct $v$, correct nodes $v\neq \ell$ send messages to $\ell$ (if $\ell\neq \bot$) or faulty nodes only, i.e., $O(n(f+1)\log |\mathcal{V}|)$ bits are sent by correct nodes.
\end{lemma}
\begin{proof}
We claim that the following algorithm achieves the stated guarantees.

\begin{algorithm2e}[H]
	\caption{Graded king consensus.}
	\KwIn{$(x_v,\ell_v)\in \mathcal{V}\times (V\,\dot{\cup}\,\{\bot\})$}
	\KwOut{$(y_v,g_v)\in \mathcal{V}\times \{0,1\}$}
	\lIf{$\ell_v\neq \bot$}{send $x_v$ to $\ell_v$\tcp*[f]{first round}}
	\If{$\ell_v=v$ and received $n-t$ times $x_v$}{
		send $x_v$ to all nodes\tcp*{second round}
		$g_v:=1$
	}
	\lElse{$g_v:=0$}
	\lIf{received $x_{\ell_v}\neq x_v$ from $\ell_v\neq \bot$}{send $x_{\ell_v}$ to all nodes\tcp*[f]{third round}}
	\lIf{received $x_w=x_v$ from $w$}{send $x_v$ to $w$\tcp*[f]{fourth round}}
	\leIf{received more than $t$ times $x$}{$y_v:=x$}{$y_v:=x_v$\tcp*[f]{break ties arbitrarily}}
\end{algorithm2e}
The running time and message size bound can be readily verified from this description. Moreover, if there is $(x,\ell)\in \mathcal{V}\times(V\cup \{\bot\})$ so that $(x_v,\ell_v)=(x,\ell)$ for all correct $v$, then $\ell$ receives $x=x_{\ell}$ at least $n-t$ times in the first round. Hence, it sends $x=x_v$ to each node $v$ in the second, so that correct nodes $v$ do not send messages in the third round. The same applies if $\ell=\bot$, as then no node $v$ receives a message from $\ell_v\neq \bot$ in the second round. Thus, correct nodes send no messages to correct nodes in the last round. Accordingly, in the final round no messages are sent to correct nodes. Overall, any message is either sent or received by $\ell$ or a faulty node, implying the second bound on communication.

\textbf{Validity:} If there is $x\in \mathcal{V}$ so that $x_v=x$ for all correct $v$, then no node $v$ can receive more than $t$ times a value $x'\neq x$ in the final round. Hence, each correct node $v$ outputs $y_v=x_v=x$.

\textbf{King Validity:} If in addition $\ell_v=\ell$ for all $v$ and correct $\ell\in V$, then $\ell$ receives $x_{\ell}=x$ from all correct nodes in the first step. Hence, it sets $g_{\ell}:=1$ and eventually outputs this value.

\textbf{Graded King Agreement}: Suppose that $\ell_v=\ell$ for all correct $v$ and some correct $\ell$, and also that $g_{\ell}=1$. Thus, $\ell$ receives $n-t$ times $x_{\ell}$ in the first round, meaning that at least $n-2t>t$ correct nodes have input $x_v=x_{\ell}$. Node $\ell$ sends $x_{\ell}$ to all nodes in the second round, prompting each $v$ with $x_v\neq x_{\ell}$ to send $x_{\ell}$ to all nodes in the third round. Each correct node $w$ with $x_w=x_{\ell}$ responds with $x_{\ell}$ in the last round, while correct nodes send no messages different from $x_{\ell}$. Hence, in the final round each $v$ with $x_v\neq x_{\ell}$ will receive more than $t$ times $x_{\ell}$, while no node receives more than $t$ times any other value. We conclude that all correct nodes $v$ output $y_v=x_{\ell}$.
\end{proof}

We can now plug these two subroutines together as follows. First, graded king consensus either ensures that all nodes agree (output $g_{\ell}=1$) or the king tells all nodes to use input $s_v=1$ for weak graded agreement, making it behave like a ``regular'' graded agreement subroutine. This establishes the requirements of king consensus while preventing high communication cost if all nodes agree on a correct leader and the same input value: in this case, graded king consensus is cheap with output $g_{\ell}=1$, which will result in cheap weak graded agreement (input $s_v=0$ at all nodes $v$).
\begin{theorem}\label{thm:king}
King Consensus can be solved in $8$ rounds with message size $O(\log |\mathcal{V}|)$. In addition, correct nodes send only $O(n(f+1)\log |\mathcal{V}|)$ bits in total if there is $(x,\ell)\in \mathcal{V}\times (V\cup \{\bot\})$ such that $(x_v,\ell_v)=(x,\ell)$ for each $v$ and $\ell$ is not faulty.
\end{theorem}
\begin{proof}
We claim that the following algorithm achieves the stated guarantees.

\begin{algorithm2e}[H]
	\caption{King consensus algorithm with improved bit complexity.}
	\KwIn{$(x_v,\ell_v)\in \mathcal{V}\times (V\,\dot{\cup}\,\{\bot\})$}
	\KwOut{$y_v\in \mathcal{V}\,\dot{\cup}\,\{\bot\}$}
	run the protocol from \Cref{lem:graded_king} with input $(x_v,\ell_v)$ and output $(k_v,h_v)$\tcp*{four rounds}
	\lIf{$\ell_v=v$ and $h_v=0$}{send ``runGC'' to all nodes\tcp*[f]{fifth round}}
	\If{received ``runGC'' from $\ell_v\neq \bot$}{
		run the protocol from \Cref{lem:weak_graded} with input $(k_v,1)$ and output $(z_v,g_v)$
	}
	\lElse{run the protocol from \Cref{lem:weak_graded} with input $(k_v,0)$ and output $(z_v,g_v)$\tcp*[f]{two rounds}}
	\lIf{$\ell_v=v$}{send $z_v$ to all nodes\tcp*[f]{eighth round}}
	\lIf{$\ell_v=\bot$}{$y_v:=\bot$}
	\lElseIf{$g_v=0$ and received $z_{\ell_v}$ from $\ell_v$}{$y_v:=z_{\ell_v}$}
	\lElse{$y_v:=z_v$}
\end{algorithm2e}
The round complexity and bound on the message size are immediate from the description and \Cref{lem:graded_king,lem:weak_graded}. If there is $(x,\ell)\in \mathcal{V}\times (V\cup \{\bot\})$ such that $(x_v,\ell_v)=(x,\ell)$ for each $v$ and $\ell$ is not faulty, \Cref{lem:graded_king} states that correct nodes send $O(n(f+1)\log |\mathcal{V}|)$ in the first step. Moreover, if $\ell\neq \bot$, in this case the king validity property of king consensus entails that $g_{\ell}=1$. Hence, no node $v$ receives a message from $\ell_v\neq \bot$ in the fifth round. Hence, all nodes $v$ use input $s_v=0$ in the call to weak graded agreement. By \Cref{lem:weak_graded}, then no messages are sent by correct nodes in the third step. In the final round, the only correct node sending messages is $\ell$. Summing up, this establishes the claimed bound on communication.

\textbf{Validity}: If there is $x\in \mathcal{V}$ so that $x_v=x$ for all correct $v$, then correct nodes output $y_v=x$ from the call to graded king consensus by its validity condition. By the validity condition of weak graded agreement, each correct node then outputs $(x,1)$ from the call to weak graded agreement. Hence, each node $v$ outputs $y_v\in \{x,\bot\}$.

\textbf{Default}: Immediately follows from the output instruction.

\textbf{King Agreement}: If there is correct $\ell \in V$ so that $\ell_v=\ell$ for all correct $v$, we distinguish two cases. If $\ell$ outputs $(k_{\ell},1)$ from the call to graded king consensus, then by graded king agreement each correct $v$ outputs $(k_{\ell},h_v)$ for some $h_v\in \{0,1\}$ from the call. By the validity condition of weak graded agreement, each correct node then outputs $(k_{\ell},1)$ from the call to weak graded agreement. As $\ell_v=\ell\neq \bot$ for each correct $v$, each $v$ then outputs $k_{\ell}$.

On the other hand, if $\ell$ outputs $(k_{\ell},0)$ from the call to graded king consensus, it sends  ``runGC'' to all nodes, which then input $s_v=1$ to the weak graded agreement instance. Again, we distinguish two cases. If some correct node $w$ outputs $g_w=1$, by weak graded agreement all nodes $v$ have $z_v=z_w$. In particular, $\ell$ sends $z_{\ell}=z_w$ in the final communication round, resulting in all nodes outputting $z_w$. Otherwise, all correct nodes $v$ have $g_v=0$ and therefore output the value $z_{\ell}$ sent by $\ell$ in the final communication round.
\end{proof}

\subsection{Filtering with less Communication}
\begin{algorithm2e}[t!]
	\caption{$(C,X,T)$-filtering algorithm with reduced amortized bit complexity, code at node $v\in V$ executed in each round. W.l.o.g., we assume that $T=\{1,\ldots,|T|\}$.\label{alg:filter}}
	\KwIn{$C_v\in [C]$ iff $v\in T$}
	\KwVars{$\hat{F}_{v,w}\in [C]$ and $s_{v,w}\in \{0,1\}$ for $w\in V$, $\hat{F}_v\in [C]$, $X_v\in [X+2]$, $N_v\in [n]$, $T_v\in [|T|]$, $a_v\in \{0,1\}$}
	\KwOut{$F_v\in [C]\cup \{\bot\}$}
	$X_v:=\max\{X_v-1,0\}$\;
	\If{$\exists c\in [C]\setminus \hat{F}_v \colon |\{w\in T\,|\, \hat{F}_{v,w}=c\}|>|T|/2$}{
		$X_v:=X+1$\;
		$a_v:=1$\;
		$\hat{F}_v:=c$\;\label{line:change}
	}
	\If{$|\{w\in V\,|\, \hat{F}_{v,w}\neq \hat{F}_v\}|>n/3$\label{line:reset}}{
		$X_v:=X+1$\;
		$a_v:=1$\;
	}
	\lIf{$T\ni w=T_v-2\bmod |T|\wedge \hat{F}_{v,w}\neq \hat{F}_v$}{$a_v:=1$\label{line:check}}
	\If{$v\in T$ and $C_v\neq \hat{F}_v+1\bmod C$}{
		$X_v:=X+1$\;
		$a_v:=1$\;
		$\hat{F}_v:=C_v$\;\label{line:change_T}
	}
	\lElse{$\hat{F}_v:=\hat{F}_v+1\bmod C$\label{line:plusone}}
	$N_v:=N_v+1\bmod n$\;\label{line:N_v}
	$T_v:=T_v+1\bmod |T|$\;\label{line:T_v}
	\ForEach{$w\in V$}{
		$\hat{F}_{v,w}:=\hat{F}_{v,w}+1\bmod C$\;\label{line:est_plusone}
		\If{$s_{v,w}=1$}{
			send $\hat{F}_v$ to $w$\;
			$s_{v,w}:=0$\;
		}
		\If{$a_v=1\vee \hat{F}_{v,w}\neq \hat{F}_v\vee w=N_v\bmod n \vee T\ni w=T_v\bmod |T|$}{
			send ``req'' to $w$\tcp*{unique bit string different from all $c\in [C]$}
			send $\hat{F}_v$ to $w$\;
		}
		\lIf{received $c\in [C]$ from $w$}{$\hat{F}_{v,w}:=C$}
		\lIf{received ``req'' from $w$}{$s_{v,w}:=1$}
	}
	$a_v:=0$\;
	\leIf{$X_v=0$}{$F_v:=\hat{F}_v$}{$F_v:=\bot$}
\end{algorithm2e}

Filtering we approach by accepting to work with possibly stale information on other nodes' view of the filtered clock. Only when inconsistencies are observed, more communication is applied. First, we can adopt a simple round-robin solution to checking whether nodes in $T$ claim a different clock value. If this value is different from its own, node $v$ will query all nodes for their values to ``refresh'' its view. Since at most $f$ nodes in $T$ do so after stabilization, this results in an amortized bit complexity of $\tilde{O}(n^2(f+1)/|T|)$, which is within budget for $|T|=\Omega(n)$. At the same time, within $f+O(1)$ rounds, $v$ will detect if its own clock value differs from the majority value in $T$ and adjust~it.

To guarantee the $X$-round crusader agreement property in case too many nodes in $T$ are faulty, $v$ also memorizes which values other nodes believe the majority in $T$ claims. Note that in this case, $|T|\in \Omega(n)$ implies that we can afford to wait until $v$ has refreshed its memorized values for all nodes in a round-robin fashion, so it is sufficient to inform other nodes whenever changing the own clock value to maintain that all views are accurate once this has been established. Thus, resetting a cooldown counter (of maximum value $X+1$) if (i) the own clock value observed as majority in $T$ changes or (ii) $t+1$ or more memorized values differ from the own guarantees that it is safe to output the own memorized value when the cooldown counter reaches $0$; otherwise, the output is $\bot$. The resulting pseudocode is given in \Cref{alg:filter}.

First, we prove that the memorized values will be accurate at the latest after $n$ rounds, but at most $f+O(1)$ rounds after $T$ counts in case $f<|T|/2$.
\begin{lemma}\label{lem:track}
\Cref{alg:filter} satisfies the invariant that $\hat{F}_{w,v,r}=\hat{F}_{v,r}$ for all correct $v,w\in V$ and rounds $r\geq n$. If $f<|T|/2$ and the variables $C_v$ count, then the invariant holds from round $f+4$.
\end{lemma}
\begin{proof}
We claim that once the invariant is established for correct $v,w\in V$, it is maintained. To see this, observe that the only lines which change $\hat{F}_v$ are Lines~\ref{line:change}, \ref{line:change_T}, and \ref{line:plusone}. If Line~\ref{line:change} or Line~\ref{line:change_T} are executed by $v$ in round $r$, $a_v$ is set to $1$ and $\hat{F}_{v,r}$ is sent to $w$ later in round $r$, i.e., once the variable $\hat{F}_v$ equals $\hat{F}_{v,r}$. On reception of this message, $w$ sets $\hat{F}_{w,v}$ to $\hat{F}_{v,r}$ and does not change this variable again in round $r$, i.e., $\hat{F}_{w,v,r}=\hat{F}_{v,r}$.

On the other hand, if $v$ executes Line~\ref{line:plusone} in round $r$, there are two possibilities. If $v$ sends a message with $\hat{F}_{v,r}$ to $w$ in round $r$, we analogously get that $\hat{F}_{w,v,r}=\hat{F}_{v,r}$. If $v$ does not send such a message, the change to $\hat{F}_{w,v}$ in round $r$ is caused by Line~\ref{line:est_plusone}. Therefore, if $\hat{F}_{w,v,r-1}=\hat{F}_{v,r-1}$, it follows that
\begin{equation*}
\hat{F}_{w,v,r}=\hat{F}_{w,v,r-1}+1\bmod C=\hat{F}_{v,r-1}+1\bmod C=\hat{F}_{v,r}.
\end{equation*}

It remains to show that this equality is established sufficiently quickly. Note that $N_v$ counts modulo $n$, because the only line affecting it is Line~\ref{line:N_v}. Hence, there is a round $r\le n$ such that $D_{v,r}=w$, prompting $v$ to send a message with $\hat{F}_{v,r}$ to $w$.

Finally, we show the stronger bound of $r_0\in O(f+1)$ provided that $f<|T|/2$ and the variables $C_u$, $u\in T$, count. We distinguish three cases.
\begin{enumerate}
  \item If $\hat{F}_{w,v,1}\neq \hat{F}_{w,1}$, $w$ will send ``req'' to $v$ in round $1$, causing it to set $s_{v,w}$ to $1$ and send $\hat{F}_{v,2}$ to $w$ in round $2$. Hence, $w$ sets $\hat{F}_{w,v,2}$ to $\hat{F}_{v,2}$.
  \item If $\hat{F}_{v,1}=\hat{F}_{w,1}$ and the previous case does not apply, then
  \begin{equation*}
  \hat{F}_{w,v,1}=\hat{F}_{w,1}=\hat{F}_{v,1}.
  \end{equation*}
  \item Suppose that the previous cases do not apply. We claim that $v$ or $w$ will set $a_v$ or $a_w$, respectively, to $1$ by round $f+3$. As we have seen, this will establish the invariant by round $f+4$. To show the claim, we fix a correct node $u\in T$. Because $\hat{F}_{v,1}\neq \hat{F}_{w,1}$, one of them must be different from $C_{u,1}$. W.l.o.g., suppose that $\hat{F}_{v,1}\neq C_{u,1}$. In particular, $v\notin T$, as otherwise Line~\ref{line:change_T} would force $v$ to set $\hat{F}_{v,1}$ to $C_{v,1}=C_{u,1}$.
  
Consider the minimal round $r\ge 1$ such that $w=T_{v,r}\bmod |T|$ is correct. As the only instruction affection $T_v$ is counting up by one module $|T|$ in each round, we have that $r\le f+1$. In this round, $v$ sends ``req'' to $w$, which responds with $\hat{F}_{w,r+1}$ in round $r+1$. As $w\in T$ is correct, it maintains that $\hat{F}_{w,r+1}=C_{w,r+1}$, and because the $A$ variables count, we have that $C_{w,r+1}=C_{u,r+1}=C_{u,1}+r$. Thus, either $v$ executed Line~\ref{line:change} in some round $r'\le r+1$ (and set $a_v$ to $1$ in this round), or it executed Line~\ref{line:plusone} in rounds $r'\le r+1$, yielding that $\hat{F}_{v,r+1}=\hat{F}_{v,1}+r\bmod C\neq C_{u,1}+r\bmod C$. We conclude that in round $r+2$, $v$ will execute Line~\ref{line:check} if it did not do so before, and set $a_v$ to $1$.\qedhere
\end{enumerate}
\end{proof}
With this invariant established, it is straightforward to prove the validity property of filtering.

\begin{lemma}\label{lem:filter_validity}
\Cref{alg:filter} satisfies the validity property of $(C,X,T)$-filtering in rounds $r\ge f+O(1)$.
\end{lemma}
\begin{proof}
The validity property of $(C,X,T)$-filtering is trivially satisfied if $f\ge |T|/2$ or the variables $C_v$, $v\in T$, do not count. Hence, assume that $f<|T|/2$ and the variables $C_v$ count. By \Cref{lem:track}, \Cref{alg:filter} then satisfies the invariant that $\hat{F}_{w,v,r}=\hat{F}_{v,r}$ for all correct $v,w\in V$ and rounds $r\ge f+4$. Moreover, each node $v\in T$ maintains that $\hat{F}_{v,r}=C_{v,r}$ due to Lines~\ref{line:change_T} and \ref{line:plusone}. As we assume that the majority of nodes in $T$ is correct and the variables $C_v$ count, this means that any $v\notin T$ executes Line~\ref{line:change} in round $r\ge f+4$ if and only if its variable $\hat{F}_{v,r}$ does not match the count of the $C_v$ variables. As the only other possibility is that the count matches and $v$ executes Line~\ref{line:plusone}, it follows that all correct nodes adopt and maintain the correct count in their variables $\hat{F}_{v,r}$ by round $f+4$. Once this took place, no variable $X_v$ can be set to $X+1$ again, implying that all correct nodes output the correct count by round $f+X+5\in f+O(1)$.
\end{proof}
If $f\ge |T|/2\in \Omega(n)$, it is sufficient to achieve $X$-round crusader agreement within $O(n)$ rounds. Otherwise, we can wait until the counting routine running on $T$ stabilizes, after which the invariant is established within $O(f+1)$ rounds. In either case, it is easy to show that the $X$-round crusader agreement property is quickly established once the invariant holds, i.e., all correct nodes have correct views of other correct nodes' $\hat{F}$ variables.
\begin{lemma}\label{lem:filter_agreement}
Suppose that $f\ge |T|/2\in \Omega(n)$ or the variables $C_v$, $v\in T$, count. Then \Cref{alg:filter} satisfies the $X$-round crusader agreement property of $(C,X,T)$-filtering in rounds $r\ge O(f+1)$.
\end{lemma}
\begin{proof}
By \Cref{lem:track}, \Cref{alg:filter} satisfies the invariant that $\hat{F}_{w,v,r}=\hat{F}_{v,r}$ for all correct $v,w\in V$ and rounds $r\ge n$. If $f\ge |T|/2$, by the assumptions of the lemma it holds that $n\in O(f)$. On the other hand, if $f<|T|/2$, the variables $C_v$ count and \Cref{lem:track} states that the invariant holds from round $f+4$. Either way, there is $r_0\in O(f+1)$ such that the invariant holds in rounds $r\ge r_0$.

Now suppose that correct node $v$ outputs $F_{v,r}=c\neq \bot$ in round $r>r_0$. Thus, it did not set $X_v$ to $X+1$ in round $r$, entailing that\footnote{Observe that the check in Line~\ref{line:reset} is performed before the variables $\hat{F}_v$ and $\hat{F}_{v,w}$ are updated in round $r$, so that state is identical to that of the end of round $r-1\ge r_0$.} $|\{w\in V\,|\, \hat{F}_{v,w,r-1}\neq \hat{F}_{v,r-1}\}|\le \lceil n/3\rceil-1=t$ due to the condition in Line~\ref{line:reset}. Thus, at least $n-2t>t$ correct nodes $w$ satisfy that $\hat{F}_{w,r-1}=\hat{F}_{v,r-1}$. Therefore, any node $w$ with $\hat{F}_{w,r-1}\neq \hat{F}_{v,r-1}$ will satisfy the condition in Line~\ref{line:reset} and set $X_w$ to $X+1$ in round $r$. As it takes at least until round $r+X+1$ for $F_w$ to decrease to $0$, the output instruction ensures that $F_{w,r'}=\bot$ for all correct $w$ with $\hat{F}_{w,r-1}\neq F_{v,r-1}$ and rounds $r'\in [r,r+X]$.

On the other hand, any correct $w$ with $\hat{F}_{w,r-1}=F_{v,r-1}$ will increase $\hat{F}_w$ by $1$ modulo $C$ in each round $r'\in [r,r+X]$ or set $X_w$ to $X+1$. Again, this will ensure that in the subsequent $X$ rounds, $w$ will output $\bot$. Therefore, $F_{w,r'}\in \{F_{v,r}+r'-r\bmod C,\bot\}$ for all correct $w$ with $\hat{F}_{w,r-1}=F_{v,r-1}$.
\end{proof}

It remains to show that the amortized bit complexity is indeed small.
\begin{lemma}\label{lem:filter_complexity}
When executing \Cref{alg:filter}, correct nodes send messages of size $O(\log C)$. If $f<|T|/2\in \Omega(n)$ and the variables $C_v$, $v\in T$, count, the bit complexity is $O(n(f+1)\log C)$ amortized over $n$ rounds.
\end{lemma}
\begin{proof}
The bound on the message size can be readily verified from the code. If $f<|T|/2$ and the variables $C_v$ count, by \Cref{lem:filter_agreement}, the output variables $F_v$, $v\in V$, count from round $r_0\in O(f+1)$. Thus, the same holds for the variables $\hat{F}_v$, $v\in V$, and by \Cref{lem:track}, we have that $\hat{F}_{v,w,r}=\hat{F}_{v,r}$ for correct $v,w\in V$. Let us check the conditions under which node $v$ sets $a_v$ to $1$ in such a round $r$:
\begin{enumerate}
  \item $\exists c\in [C]\setminus \hat{F}_v \colon |\{w\in T\,|\, \hat{F}_{v,w}=c\}|>|T|/2$. As $f<|T|/2$ and all variables of correct nodes agree, this cannot happen.
  \item $|\{w\in V\,|\, \hat{F}_{v,w}\neq \hat{F}_v\}|>n/3$. As $f\le t<n/3$ and all variables of correct nodes agree, this cannot happen.
  \item $T\ni w=T_{v,r}-2\bmod |T|$ and $\hat{F}_{v,w}\neq \hat{F}_v$. As correct nodes' variables agree, this is possible only if $w$ is faulty. As $f<|T|/2\in \Omega(n)$ and $T_v$ counts modulo $|T|$, this occurs at a rate of $O(f/n)$ amortized over $n$ rounds.
  \item $v\in T$ and $C_v\neq \hat{F}_v+1\bmod C$. As the variables $C_v$ count and match $\hat{F}_v$ for correct $v\in T$, this cannot happen.
\end{enumerate}
Now let us sum up all messages sent by correct nodes in rounds $r>r_0$. Consider the possible causes for a correct node $v$ to send a message to correct node $w$.
\begin{itemize}
  \item $a_v=1$. As we have observed, the amortized rate over $n$ rounds at which this happens is $O(f)$ when summing over all nodes. In a round in which $a_v$ is $1$, $v$ sends $O(n \log C)$ bits of communication. Overall, this contributes an amortized bit complexity of $O(nf\log C)$.
  \item $\hat{F}_{v,w}\neq \hat{F}_v$. Since correct nodes' variables agree, this does not occur.
  \item $w=N_v\bmod n$ or $T\ni w=T_v\bmod |T|$. Since these are two nodes only, the corresponding messages contribute a total of $O(n\log C)$ to the bit complexity.
  \item $s_{v,w}=1$. As $s_{v,w}$ is set to $0$ again after the triggered message is sent, it must have set to $1$ in round $r-1$. Thus, any such message can be attributed to one of the preceding causes without affecting the asymptotic amortized bit complexity.
\end{itemize}
Finally, note that, trivially, correct nodes send at most $nf$ messages to faulty nodes in each round. Summing up all terms, we arrive at the stated bound of $O(n(f+1)\log C)$ on the amortized bit complexity.
\end{proof}
It is worth noting that, if $f$ is small, \Cref{lem:filter_complexity} only controls the bit complexity once the variables $C_v$ count. Hence, we need to take into account the stabilization time of the utilized counting algorithm when amortizing the bit complexity, resulting in the following corollary.
\begin{corollary}\label{cor:filtering}
Assume that $|T|\in \Omega(n)$. Then \Cref{alg:filter} solves $(C,X,T)$-filtering with stabilization time $O(f+1)$ and messages of size $O(\log C)$. If we assume that $f<|T|/2$ implies that the variables $C_v$ start counting within $k$ rounds, its bit complexity amortized over $kn/(f+1)$ rounds is $O(n(f+1)\log C)$.
\end{corollary}
\begin{proof}
Follows from \Cref{lem:filter_validity,lem:filter_agreement,lem:filter_complexity}, noting that if the variables $C_v$ start counting only after $k$ rounds, we need to amortize the resulting $O(kn^2\log C)$ bits of communication over at least $kn/(f+1)$ rounds to not affect the asymptotic amortized bit complexity.
\end{proof}
We reach the following intermediate result, which meets our goals up to the additive $O(\log n)$ overhead in stabilization time.
\begin{theorem}
$C$-counting can be solved with stabilization time $O(f+1+\log n)$ and message size $O(\log C + \log^2 n)$. Moreover, the bit complexity amortized over $n\log n$ rounds is $O(n((\log n+\log C)f+\log^2 n))$.
\end{theorem}
\begin{proof}
We recursively apply \Cref{alg:recursion} using \Cref{alg:filter} for filtering and the king consensus algorithm from \Cref{thm:king}, where we partition the node set as evenly as possible in each recursion step and the base case of $n=1$ is trivial. Due to the even partition in each recursive step, we can apply \Cref{cor:filtering} and reason analogously to \Cref{cor:stab} to infer that the resulting algorithm solves $C$-counting with stabilization time $O(f+1+\log n)$.
The bound on message size is immediate from the fact that the recursion has depth $O(\log n)$ and the bounds on message size from \Cref{thm:king,cor:filtering}.

It remains to bound the amortized bit complexity. To this end, consider the recursion tree, where we label each recursive instance of \Cref{alg:recursion} by the node set on which it runs. We charge each node with (bounds on) the amortized communication complexity of its filtering and king consensus instances, but not its recursive calls. We will then sum over all nodes to bound the total amortized communication complexity.

Note that the leaves of the tree correspond to trivial instances on single nodes, which do not contribute to communication costs. For an inner node of the tree corresponding to an instance of \Cref{alg:recursion} on node set $V'$ with $f'<|V'|/3$ faulty nodes, by the already estabilished bound on the stabilization time we have the following.
\begin{itemize}
  \item We can apply \Cref{cor:filtering} with $k=O(f'+1+\log |V'|)=O(f'+1+\log n)$ to bound the contribution of filtering to the amortized bit complexity over $n\log n$ rounds by $O(|V'|(f'+1)\log |V'|)=O(|V'|(f'+1)\log n)$.
  \item If $f'<\lfloor |V'|\rfloor/9$, by \Cref{lem:filter_good} the variables $F_v^b$ of the instance of \Cref{alg:recursion} on $V'$ start to count within $O(f'+1+\log n)$ rounds for both $b\in \{0,1\}$. From then on, over $\Theta(n)$ rounds there are only $O(f')$ instances of king consensus to which correct nodes input a faulty leader $\ell$. Therefore, \Cref{thm:king} states that the total bit complexity of king consensus instances on node set $V'$ within $\Theta(n)$ rounds is $O(|V'|(f'+1)\log |V'|)$ (except for the initial call on node set $V$ corresponding to the root of the tree, where message size is $O(\log C)$). Amortizing over $n\log n$ rounds, the initial cost of $O(n^2(f'+1+\log n))$ messages before the variables $F_v^b$ start counting does not increase the amortized bit complexity beyond $O(|V'|(f'+1)\log |V'|)=O(|V'|(f'+1)\log n)$ (or $O(n(f+1)\log C)$ for the root of the tree).
\end{itemize}
Observe that if $f'\ge |V'|/3$, no stabilization is guaranteed. However, in this case the bounds on message size still apply and $f'=\Omega(|V'|)$. Therefore, in this case the contribution of the node to the amortized bit complexity is $O(|V'|^2\log n)=O(|V'|(f'+1)\log n)$. In summary, this bound applies uniformly to all nodes of the recursion tree.

Now consider all nodes in depth $d$ of the recursion tree. As \Cref{alg:recursion} is always applied with $V_0$ and $V_1$ partitioning the node set, the union of all nodes in calls in depth $d$ is a subset of $V$. Moreover, because each recursive call partitions as evenly as possible, $|V'|\le (2/3)^d n$. Hence, summing over all nodes in depth $d$, we get a contribution of $O((f(2/3)^d n+n)\log n)$. Summing over all $O(\log n)$ values of $d$ for which the tree contains nodes in depth $d$, we arrive at a total bit complexity amortized over $n\log n$ rounds of $O(n((\log n+\log C)f+\log^2 n))$.
\end{proof}

%% file: opt_stab.tex
\section{Asymptotically Optimal Stabilization Time}\label{sec:opt_stab}
The additive overhead of $O(\log n)$ in the stabilization time is owed to the recursion depth of $O(\log n)$. Even if there are no faults, the stabilization of the counts needs to propagate from a leaf of the recursion tree to the root, costing $O(1)$ rounds for each level of the recursion. Hence, to reach stabilization time $O(f+1)$, we need a way to stabilize within $O(1)$ rounds if there are no faults. We achieve this by means of a different recursive template, which has one branch ensure this property, while the other simply ``recurses'' on the entire node set and applies \Cref{alg:recursion}. The challenge is that we cannot amortize the cost incurred by faulty leaders as in \Cref{sec:communication}, as the single node providing the ``recursively'' constructed clock would control which nodes act as leaders. This is problematic if $f$ is small, as we then do not have the budget to cover repeated calls to king consensus with $\Omega(n^2)$ cost.

\subsection{Weak King Consensus}

As our way out, we relax the king consensus task such that king agreement is only required when \emph{all} nodes are correct. For our use case of $f=0$, this relaxation is uncritical.

\begin{definition}[Weak King Consensus]
The \emph{Weak King Consensus} problem is specified as follows.
Each node $v$ has inputs $x_v\in \mathcal{V}$ and $\ell_v\in V\,\dot{\cup}\,\{\bot\}$.
Each node $v$ computes an output $y_v\in \mathcal{V}\,\dot{\cup}\,\{\bot\}$ with the following guarantees:
\begin{itemize}
 \item \textbf{Validity}: If there is $x\in \mathcal{V}$ so that $x_v=x$ for all correct $v$, then $y_v\in \{x,\bot\}$ for all correct $v$.
 \item \textbf{Default}: If $\ell_v=\bot$ for all correct $v$, then $y_v=\bot$ for all correct $v$.
 \item \textbf{Weak King Agreement}: If $f=0$ and there is $\ell \in V$ so that $\ell_v=\ell$ for all $v$, then $y_v=y_{\ell}\neq \bot$ for all $v$.
\end{itemize}
\end{definition}
To obtain a cheap implementation of weak king consensus, we make use of edge expanders. An edge expander guarantees that all node sets of size at most $n/2$ have an $\varepsilon$-fraction of their incident edges leave the set.
\begin{definition}[Edge Expander]
The \emph{edge expansion} of a graph $G=(V,E)$ is defined as
\begin{equation*}
h(G):=\min_{0<|S|\le n/2, S\subset V}\left\{\frac{|E\cap (S\times V\setminus S)|}{|S|}\right\}.
\end{equation*}
\end{definition}
\begin{theorem}[\cite{gabber79explicit}]
There is a family of constant-degree graphs $G_n=(\{1,\ldots,n\},E_n)$ of constant expansion, i.e., $h(G_n)\ge \varepsilon>0$ for $\varepsilon\in \Omega(1)$, such that the neighbors of any given node can be computed efficiently, i.e., with $\log^{O(1)} n$ computational steps.
\end{theorem}
Using a constant-degree expander, checking whether expander neighbors have differing values will not cause more than $O(f)$ nodes to observe a difference after stabilization. On the other hand, if there are no faults, the expander ensures that the number of nodes seeing a difference in case $f=0$ is proportional to the number of nodes with a minority value. This provides proof that more communication is justified---either, because stabilization has not occurred yet or because there are many faults. This is sufficient to convince all nodes with minority value to query all nodes for their values, which then proves to them that they can safely listen to the leader and adopt the majority value it proposes. Crucially, all of this only needs to work if $f=0$; only validity must be maintained under all circumstances.
\begin{theorem}\label{thm:weak_king}
Weak king consensus can be solved in $6$ rounds with messages of size $O(\log |\mathcal{V}|)$. If there is $x\in \mathcal{V}$ such that $x_v=x$ for all correct nodes $v$, then the bit complexity is $O(n(f+1)\log |\mathcal{V}|)$.
\end{theorem}
\begin{proof}
We claim that the following algorithm achieves the stated guarantees.

\begin{algorithm2e}[H]
	\caption{Weak king consensus. The parameter $0<\varepsilon\in \Omega(1)$ is chosen such that graphs of expansion $\varepsilon$ can be constructed efficiently.}
	\KwIn{$(x_v,\ell_v)\in \mathcal{V}\times (V\,\dot{\cup}\,\{\bot\})$}
	\KwOut{$y_v\in \mathcal{V}\,\dot{\cup}\,\{\bot\}$}
	compute $\{v,w\}\in E$ for a constant-degree graph $(V,E)$ of expansion $\varepsilon>0$\;
	send $x_v$ to all $w$ with $\{v,w\}\in E$ and to $\ell_v$\tcp*{first round}
	\lIf{received $x_w\neq x_v$ from $w$ and $\{v,w\}\in E$}{send ``alert'' to all nodes\tcp*[f]{second round}}
	\lIf{received $k_v\neq 0$ times ``alert'' and $w=v+i\bmod n$ for $i\in [0,\lceil 2k_v/\varepsilon\rceil]$}{send ``req'' to $w$\tcp*[f]{third round}}
	\lIf{received ``req'' from $w$}{send $x_v$ to $w$\tcp*[f]{fourth round}}
	\If{received $x_w\neq x_v$ from half of the nodes $w=v+i\bmod n$ with $i\in [0,\lceil 2k_v/\varepsilon\rceil]$}{
		send ``req'' to all nodes\tcp*{fifth round}
	}
	\lIf{$\ell_v\neq v$ and $v$ received a ``req'' message from node $w$}{send $x_v$ to $w$}
	\ElseIf{$\ell_v=v$ and $v$ received messages from all nodes in the first round}{
		send a value $z_v$ most frequently received in the first round to all nodes\tcp*{sixth round}
	}
	\lIf{$\ell_v=\bot$}{$y_v=\bot$}
	\lElseIf{received $t+1$ values different from $x_v$ and $z_{\ell_v}$ from $\ell_v$}{$y_v:=z_{\ell_v}$}
	\lElse{$y_v:=x_v$}
\end{algorithm2e}
The bounds on the number of rounds and the message size are immediate from the description of the algorithm. To bound the bit complexity if there is $x\in \mathcal{V}$ such that $x_v=x$ for all correct nodes $v$, we control the ammount of communication in reach round.
\begin{enumerate}
  \item As $(V,E)$ has constant degree, correct nodes send $O(n)$ messages.
  \item As $(V,E)$ has constant degree and $x_v=x$ for all correct $v$, $O(f)$ nodes send a total of $O(nf)$ ``alert'' messages.
  \item As each node receives $k_v\in O(f)$ ``alert'' messages, the total number of ``req'' messages in this round is $O(nf)$.
  \item As nodes only respond to ``req'' messages from the preceding round, $O(nf)$ messages are sent.
  \item Going over the nodes from $1$ to $n$, we greedily select nodes $v$ with $k_v>0$ that in the third round do not receive a ``req'' message from the preceding one that was selected, but send ``req'' messages in the fifth round. Excepting the last such node, by construction no faulty node receives a ``req'' message from more than one selected node. However, each selected node must have sent ``req'' messages to at least $k_v/\varepsilon$ faulty nodes, as $x_v=x$ for correct nodes. Hence, $f\ge \sum_{v\mbox{ selected}}k_v/(2\varepsilon)$. On the other hand, the greedy selection ensures that the union $\bigcup_{v\mbox{ selected}}\{v,\ldots,v+\lceil 2k_v/\varepsilon\rceil\}$ contains all nodes that send ``req'' messages in the fifth round. As this union has size at most $O(\sum_{v\mbox{ selected}}k_v)$, at most $O(nf)$ ``req'' messages are sent in the fifth round.
  \item The number of responses to ``req'' messages is $O(nf)$. Recall that in the first round correct nodes send only $O(1)$ messages each. Thus, there can only be $O(1)$ nodes that have received messages from all nodes. Hence, nodes $v$ with $\ell_v=v$ send $O(n)$ messages. Overall, $O(n(f+1))$ messages are sent in this round.
\end{enumerate}
Summing up over all rounds and using that messages have size $O(\log |\mathcal{V}|)$, the bound on the bit complexity follows.

\textbf{Validity:} In order for $v$ to output $y_v\notin \{x_v,\bot\}$, it must have received at least $t+1$ messages carrying values different from $x_v$ in the final round. If there is $x\in \mathcal{V}$ so that $x_w=x$ for all correct nodes $w$, then correct nodes $w$ with $\ell_w\neq w$ only ever send messages with value $x$ (or ``req'' or ``alert'' messages). Thus, a node $w$ with $\ell_w=w$ that receives messages from all nodes in round $1$ will not send $z_{\ell_w}\neq x$ in the final round either, since it received a strict majority of messages with value $x$ in round $1$. Hence, in this case it must hold that $y_v\in \{x_v,\bot\}$.

\textbf{Default:} This property is immediate from the output instruction.

\textbf{Weak King Agreement:} Suppose that $f=0$ and that there is $\ell\in V$ such that $\ell_v=\ell$ for all $v\in V$. Then $\ell$ receives messages from all nodes in the first round and sends a value $z:=z_{\ell}$ to all nodes in the final round. We claim that $y_v=z$ for all nodes $v$.

To see this, assume towards a contradiction that $v\in V$ is a node with $y_v\neq z$. As $v$ received $z$ from $y_\ell=\ell$ in the final round, it must have received $x_w\neq x_v$ from at most $t$ nodes $w$ in the final round. If $v$ sent ``req'' to all nodes in the second to last round, this is impossible, as $x_v$ is the input of at most $n/2<n-t$ nodes. Thus, $v$ did not do so, implying that it received $x_v$ from at least half of the nodes $v+i\bmod n$ with $i\in [0,\ldots,\lceil 2k_v/\varepsilon\rceil]$.

Denoting $V_{x_v}:=\{w\in V\,|\,x_w=x_v\}$, it follows that $|V_{x_v}|>k_v/\varepsilon$. However, because $x_v\neq z$ is not more frequent than $z$, $|V_{x_v}|\le n/2$. Hence, at least an $\varepsilon$-fraction of all edges of $(V,E)$ that are incident to nodes in $V_{x_v}$ has their other endpoint outside of $V_{x_y}$. As $(V,E)$ has uniform degree, it follows that at least an $\varepsilon$-fraction of the nodes in $V_{x_v}$ has a neighbor in $V\setminus V_{x_v}$. Thus, we arrive at the contradiction that at least $\varepsilon |V_{x_v}|\le k_v<\varepsilon |V_{x_v}|$ nodes send an ``alert'' message to all nodes in round $2$.
\end{proof}

\subsection{Recursion Template Ensuring \texorpdfstring{$O(1)$-round Stabilization if $f=0$}{O(1)-round Stabilization if f=0}}
As our solution to weak king consensus is guaranteed to incur a cost of only $\tilde{O}(nf)$ after stabilization, regardless of whether the leader is correct, it is safe to use in an imbalanced recursion, see \Cref{alg:recursion_0}. This template exploits that if there is no fault, the leader will succeed on \emph{any} attempt supported by all nodes, so we can let nodes ignore weak king consensus instances if the filtered clock from the recursively constructed counter on the whole node set asks them to execute a king consensus instance (with a non-$\bot$ leader).
\begin{algorithm2e}[t]
	\caption{A second recursion template for $C$-counting. We give the code $v\in V$ executes in each round, where $C_v$ is the output variable. The template is parametrized by a solutions to counting, filtering with $V$ as clock set, $R$-round king consensus, and $R$-round weak king consensus, where $R\in O(1)$, all on node set $V$. To highlight the similarities to \Cref{alg:recursion}, we use the shorthands $k:=3R$ and $X:=6R$.\label{alg:recursion_0}}
	\KwVars{for $b\in\{0,1\}$, outputs $A_v$ and $F_v$ of counting on $V$ and filtering with clock set $V$, modulo-$R$ counter $P_v\in [R]$, cooldown counter $L_v\in [R+1]$}
	\KwOut{$C_v\in [C]$}
	\leIf{$F_v=kw\bmod (kn)$}{$\ell_v:=w$}{$\ell_v:=\bot$}
	initialize king consensus on universe $[C]\cup \{\bot\}$ with input $(C_v+R\bmod C,\ell_v)$\;
	\For{$r\in \{R,\ldots,1\}$}{
		execute the code of round $r$ of king consensus on the state stored for round $r$\;
		\If{$r=R$ and the output of the instance is $c\in [C]$}{
			$C_v:=c$\;
		}
		\Else{
			store the new state in the memory block allocated for round $r+1$\;
		}
	}
	execute the code of an instance of $(k n,X,V)$-filtering with input $A_v$ and output $F_v$\;
	execute the code of an instance of $(kn)$-counting on node set $V$ with output $A_v$\;
	$C_v:=C_v+1\bmod C$\;
	\If{$v=0$}{
		$P_v:=P_v+1\bmod R$\;
		send $P_v$ to all nodes\;
	}
	\ElseIf{received $P_0$ from node $0$}{
		$P_v:=P_0$\;
	}
	\lIf{$\ell_v\neq \bot$}{$L_v:=R$}
	\leIf{$P_v=0\bmod R$ and $L_v=0$}{$\ell_v':=0$}{$\ell_v':=\bot$}
	initialize weak king consensus on universe $[C]\cup \{\bot\}$ with input $(C_v+R-1\bmod C,\ell_v')$\;
	$L_v:=L_v-1$\;
	\For{$r\in \{R,\ldots,1\}$}{
		execute the code of round $r$ of weak king consensus on the state stored for round $r$\;
		\If{$r=R$ and the output of the instance is $c\in [C]$}{
			$C_v:=c$\;
		}
		\Else{
			store the new state in the memory block allocated for round $r+1$\;
		}
	}
\end{algorithm2e}

In the following, denote by $S_C$ the stabilization time of the instance of $(k n)$-counting on node set $V$ and by $S_F$ the stabilization time of the instance of $(k n,X,V)$-filtering. Moreover, we use \Cref{def:unimpeded} in the sense that $b=0$ corresponds to instances of weak king consensus (whose only possible leader is node $0$) and $b=1$ corresponds to instances of king consensus. To prove stabilization, as in \Cref{sec:counting} it suffices to show that an unimpeded instance occurs.
\begin{lemma}\label{lem:opt_unimpeded}
There is an unimpeded instance of king consensus within $S_C+S_F+O(f+1)$ rounds.
\end{lemma}
\begin{proof}
As $f\le t<|V|/3$, within $S_C+S_F$ rounds the variables $F_v$ start to count. Thus, there is some round $r$ such that a king consensus instance running with some correct leader is initialized, and no correct node uses input $\ell_v\neq \bot$ for king consensus instances in rounds $r,r+1,\ldots,r+R$. In round $r$ all correct nodes $v$ set $L_v$ to $R$ before determining $\ell_v'$. As $L_v$ needs to decrement $R$ times to reach $0$, which happens after $\ell_v'$ is determined, no correct node inputs $\ell_v'\neq \bot$ in rounds $r,r+1,\ldots,r+R$.
\end{proof}

The above lemma handles the case that $f>0$. If $f=0$, fast stabilization is facilitated by node $0$.
\begin{lemma}\label{lem:opt_unimpeded_0}
If $f=0$, there is an unimpeded instance of weak king consensus by round $S_F+O(1)$.
\end{lemma}
\begin{proof}
If $f=0$, the variables $P_v$ $R$-count from round $1$. Analogously to \Cref{lem:filter_bad}, for rounds $r\ge S_F$ there is a virtual clock $\hat{F}$ such that $F_{v,r}\in \{\hat{F}_r,\bot\}$ and, for any $X$ consecutive rounds, $\hat{F}_{r+1}=\hat{F}_r+1\bmod (kn)$. We distinguish two cases.
\begin{enumerate}
  \item If $\hat{F}_r\bmod k\neq 0$ in rounds $r,r+1,\ldots,r+3R-2$, all nodes have $\ell_v=\bot$ in these rounds. Thus, by round $r+R-1$, the counters $L_v$ all become $0$. By round $r+2R-1$, an instance of weak king consensus is initialized, and it terminates by round $r+3R-2$. Hence, it is unimpeded.
  \item Otherwise, $\hat{F}_r\bmod k\neq 0$ in rounds $r+3R-1,r+3R,\ldots,r+6R-3$, where we use that $k=3R$ and $X=6R$. Analogously to the first case, we conclude that an unimpeded instance of weak king consensus is initialized by round $r+5R$.
\end{enumerate}
As $R=O(1)$, this concludes the proof.
\end{proof}

\begin{lemma}\label{lem:opt_stab}
\Cref{alg:recursion_0} solves $C$-counting with stabilization time $S_C+S_F+O(f+1)$. If $f=0$, the stabilization time is $S_F+O(1)$.
\end{lemma}
\begin{proof}
Analogously to the proof of \Cref{lem:stabilization}, we can show that if an unimpeded king consensus instance is initialized in round $r$, the variables $C_v$ count from round $r+R$: the king agreement property of king consensus is only used to establish that the outputs of the unimpeded instance agree, so using weak king agreement on other instances is not an issue. On the other hand, if $f=0$, then the weak king agreement and king agreement properties coincide, i.e., for $f=0$ a solution to weak king agreement is a solution to king agreement. Hence, under the constraint that $f=0$, an unimpeded instance of weak king consensus being initialized in round $r$ implies that the variables $C_v$ count from round $r+R$ as well. The statement of the lemma now follows by applying \Cref{lem:opt_unimpeded,lem:opt_unimpeded_0}.
\end{proof}

\subsection{Combining the Techniques}
It remains to merge the two recursive templates so that the length of the path in the recursion tree that guarantees stabilization is reduced from $O(\log n)$ to $O(\log f)$.
\main*
\begin{proof}
We combine our two recursive templates as follows. We alternate between
\begin{itemize}
  \item [(i)] applying the template from \Cref{alg:recursion_0}, where the counting algorithm on $V$ is given by
  \item [(ii)] applying the template from \Cref{alg:recursion} with $V_0=\{1,\ldots,\lfloor |V|/2\rfloor\}$ and $V_1=V\setminus V_0$ (choosing constants appropriately), with the counting algorithms given by (i) or being trivial if only one node remains.
\end{itemize}
We use the $8$-round solution to king consensus given by \Cref{thm:king}, the $6$-round solution to weak king consensus given by \Cref{thm:weak_king} (padded to $R=8$ rounds), and the filtering algorithm given in \Cref{alg:filter}. Note that \Cref{alg:filter} is always applied with $|T|\in \Omega(|V|)$.

Consider the recursion tree, whose nodes are labeled by the node sets $V'$ on which the respective template is applied. Observe that on inner nodes on even levels, the template from \Cref{alg:recursion_0} is applied, while on inner nodes on odd levels, the template from \Cref{alg:recursion} is used. Thus, the depth of the tree is $O(\log n)$ due to the balanced split on odd levels. As we partition the nodes on odd levels, this immediately yields the bound on the message size.

Concerning the stabilization time, consider the root-leaf path constructed by (i) on even levels, going to the child that is a leaf if and only if $V'$ contains no faulty nodes, and (ii) on odd levels to the child whose node set contains fewer faults. Along this path, the number of faults $f'$ in the current set is reduced by at least half for every odd step, and once it reaches zero, the path reaches a leaf on the next hop. In particular, the path maintains the invariant $f'<|V'|/3$. By \Cref{thm:stab,cor:filtering,lem:opt_stab}, we get at each node on the path that the stabilization time is $S_{\mathrm{next}}+O(f'+1)$, where $S_{\mathrm{next}}$ is the stabilization time of the next node on the path. We conclude that the stabilization time of the root node, i.e., the overall algorithm, is bounded by
\begin{equation*}
O\left(\sum_{d=0}^{O(\log f)}\left\lfloor 2^{-\lfloor d/2\rfloor}\right\rfloor f+1\right)=O(f+1).
\end{equation*}

With the bound on the stabilization time established, it remains to bound the amortized bit complexity. We annotate each node of the recursion tree with a bound on its contributions that are not caused by recursive calls (amortized over $n$ rounds), so that summing over all nodes of the tree results in the desired bound. Since leaves do not contribute, we need to consider inner nodes only. Denote by $V'$ and $f'$ the node set corresponding to the considered inner node and number of faults in this node set, and observe that on each level of the tree, the node sets of inner nodes are disjoint.
\begin{itemize}
  \item By \Cref{thm:weak_king}, weak king consensus instances incur a cost of $O(|V'|(f'+1)\log |\mathcal{V}|)$ if nodes agree on the inputs $x_v$. If $f'<|V'|/3$, this occurs within $O(f'+1)$ rounds, during which $O(|V'|^2\log |\mathcal{V}|)$ bits are sent. Amortizing over $n\ge |V'|$ rounds, the contribution becomes $O(|V'|(f'+1)\log |\mathcal{V}|)$ on average. If $f'\ge |V'|/3$, the message size bound from the theorem yields a bound of $O(|V'|^2\log |\mathcal{V}|)=O(|V'|(f'+1)\log |\mathcal{V}|)$.
  \item By \Cref{thm:king}, king consensus instances incur a cost of $O(|V'|(f'+1)\log |\mathcal{V}|)$ if nodes agree on the inputs $x_v$ \emph{and} non-faulty $\ell_v\in V'\cup \{\bot\}$. If $f<|V'|/9$, even on odd levels both recursive instances have few enough faults to stabilize within $O(f'+1)$ rounds. By \Cref{lem:filter_good,cor:filtering}, the variables $F_v^b$ (or $F_v$ on even levels) thus start counting within $O(f'+1)$ rounds. Also, we already know that the output variables start counting within $O(f'+1)$ rounds. Hence, barring an initial $O(f'+1)$ rounds, all instances of king consensus except for $O(f')$ many every $|V'|\le n$ rounds incur only cost $O(|V'|(f'+1)\log |\mathcal{V}|)$. Using the message size bound of $O(\log |\mathcal{V}|)$ and amortizing over $n$ rounds, the amortized bit complexity is bounded by $O(|V'|(f'+1)\log |\mathcal{V}|)$.
  \item Unless $f'=\Theta(|V'|)$, we can apply \Cref{cor:filtering} with $k=O(f'+1)$ to bound the amortized bit complexity of the instances of \Cref{alg:filter} by $O(|V'|(f'+1)\log n)$; if $f'=\Theta(|V'|)$, the same bound trivially follows from the message size bound.
  \item \Cref{alg:recursion_0} causes additional communication of $O(|V'|)$ bits per round due to node $0$ sending the modulo $R\in O(1)$ counter $P_1$ to all nodes in $V'$.
\end{itemize}
Summing over all inner nodes in a given depth $d$, using that $|V'|$ falls by factor at least $3$ for every $2$ levels, and taking into account that the considered node sets are disjoint, we get a total contribution of $O(n((2/3)^{d/2}f+1)(\log |\mathcal{V}|+\log n)$ from nodes in depth $d$, where $\log |\mathcal{V}|=O(\log n)$ in all cases but the consensus instances for $d=0$, where $\log |\mathcal{V}|=O(\log C)$. Summing over all $O(\log n)$ values of $d$, we arrive at an overall bit complexity of $O(n(f\log C+\log^2 n))$ amortized over $n$ rounds.
\end{proof}

%% file: open.tex
\section{Open Questions}\label{sec:open}
We conclude the paper by listing a number of open questions we consider to be of interest.
\begin{enumerate}
  \item Is the achieved bit complexity (nearly) optimal?
  \item Is it possible to interleave our approach with randomized techniques to get both stabilization time $O(f+1)$ deterministically and $O(1)$ in expectation?
  \item Can the proposed or similar techniques be applied to the pulse synchronization problem to achieve early stabilization also there?
  \item Cryptographic tools can enable to increase $t$ to $\lfloor (n-1)/2\rfloor$ and circumvent the Dolev-Reischuk bound, e.g.\ by means of threshold signatures~\cite{desmedt89threshold}. Can higher resilience or lower communication cost be achieved while maintaining early stabilization?
  \item Addressing the previous question requires to reconcile the total loss of volatile state and arbitrary transient violation of the protocol assumed in the context of self-stabilization with the need to hide secret keys and signatures from the adversary. Can this be achieved under reasonable assumptions? And if yes, how should this be modeled?
\end{enumerate}

%% file: app_implicit.tex
\section{Results Implicit in Prior Work}\label{app:implicit}
\filtering*
\begin{proof}
The claim that the following algorithm solves the task.

\begin{algorithm2e}[H]
	\caption{$(C,X,T)$-filtering algorithm implicit in~\cite{lenzen17counting,lenzen19counting}.}
	\KwIn{$C_v\in [C]$ iff $v\in T$}
	\KwVars{$m_v,M_v\in [C]\cup \{\bot\}$, $X_v\in [X+1]$}
	\KwOut{$F_v\in [C]\cup \{\bot\}$}
	send $(C_v,m_v)$ to all nodes\;
	\leIf{received $(C,\cdot)$ from the majority of nodes $w\in T$}{$m_v:=C$}{$m_v:=\bot$}
	\If{received $(\cdot,m)$ from at least $n-t$ nodes $w\in V$}{
		\leIf{$m=M_v+1\bmod C$}{$X_v:=\max\{X_v-1,0\}$}{$X_v:=X$}
		$M_v:=m$
	}
	\lElse{$X_v:=X$}
	\leIf{$X_v=0$}{$F_v:=M_v$}{$F_v:=\bot$}
\end{algorithm2e}
The bound on message size can be readily verified from this description.

\textbf{Validity:} If fewer than $|T|/2$ nodes in $T$ are faulty and the variables $C_v$ $C$-count, the variables $m_v$ and $M_v$ count from rounds $1$ and $2$, respectively. Then all counters $X_v$ count down to $0$ by the end of round $X+2$, so that the outputs count from round $X+2$.

\textbf{$X$-round Crusader Agreement}: If in round $r\ge X+2$ correct $v$ outputs $C_{v,r}\in [C]$, it follows that it received $M_{v,r'}=m_{w,r'-1}=C_{v,r}+r'-r\bmod C$ from at least $n-t$ nodes $w\in V$ each in rounds $r'\in \{r-X,r-X-1,\ldots,r\}$. At least $n-2t>t$ of these are correct, implying that no $w$ receives $n-t$ times another value in round $r'$. Hence, $M_{w,r'}\in \{M_{v,r'},\bot\}$ for each correct $w$ and such $r'$, implying that also $F_{w,r'}\in \{M_{v,r'},\bot\}=\{C_{v,r}+r'-r\bmod C,\bot\}$, as required.
\end{proof}

\graded*
\begin{proof}
We claim that the following algorithm achieves the stated guarantees.

\begin{algorithm2e}[H]
	\caption{Graded agreement algorithm implicit in~\cite{berman89consensus}.}
	\KwIn{$x_v\in \mathcal{V}$}
	\KwOut{$(y_v,g_v)\in \mathcal{V}\times \{0,1\}$}
	send $x_v$ to all nodes\tcp*{first communication round}
	\lIf{received $x_v$ $n-t$ times}{send $x_v$ to all nodes\tcp*[f]{second communication round}}
	\lIf{received $x_v$ $n-t$ times}{$(y_v,g_v):=(x_v,1)$}
	\lElseIf(\tcp*[f]{break ties arbitrarily}){received $t+1$ times $x\in \mathcal{V}$}{$(y_v,g_v):=(x,0)$}
	\lElse{$(y_v,g_v):=(x_v,0)$}
\end{algorithm2e}
The running time and bound on message size can be readily verified from this description.

\textbf{Validity}: If there is $x\in \mathcal{V}$ so that $x_v=x$ for all correct $v$, then each $v$ receives $n-t$ times $x$ twice and outputs $(x,1)$.

\textbf{Graded Agreement}: Suppose that $g_w=1$ for some correct $w$. Hence, $w$ received at least $n-2t>t$ messages with $y_w$ in the second round from correct nodes. These must have received at least $n-2t>t$ messages with $y_w$ from correct nodes in the first round. It follows that each correct $v$ with $x_v\neq y_w$ receives more than $t$ values different from $x_v$ in the first round, implying that it will not send $x_v$ again in the second round. Therefore, each node receives more than $t$ times $x_w$ and at most $t$ times $x$ for any $x\neq x_w$ in the second round. We conclude that each correct $v$ outputs $y_v=y_w$.
\end{proof}

\king*
\begin{proof}
We claim that the following algorithm achieves the stated guarantees.

\begin{algorithm2e}[H]
	\caption{King consensus algorithm implicit in~\cite{berman89consensus}.}
	\KwIn{$(x_v,\ell_v)\in \mathcal{V}\times (V\,\dot{\cup}\,\{\bot\})$}
	\KwOut{$y_v\in \mathcal{V}\,\dot{\cup}\,\{\bot\}$}
	run the protocol of \Cref{lem:graded} with input $x_v$; denote output by $(z_v,g_v)$\tcp*{two rounds}
	\lIf{$\ell_v=\bot$}{$y_v:=\bot$}
	\lElseIf{$g_v=0$ and received $z_{\ell_v}$ from $\ell_v$}{$y_v:=z_{\ell_v}$}
	\lElse{$y_v:=z_v$}
\end{algorithm2e}
The running time and bound on message size can be readily verified from this description.

\textbf{Validity}: If there is $x\in \mathcal{V}$ so that $x_v=x$ for all correct $v$, by validity of graded agreement all correct nodes output $(z_v,g_v)=(x,1)$ from their call to graded agreement. Hence, they all output $x$ or $\bot$, as required.

\textbf{Default}: Immediately follows from the output instruction.

\textbf{King agreement}: If there is correct $\ell \in V$ so that $\ell_v=\ell$ for all correct $v$, distinguish two cases. If some correct node $w$ outputs $(z_w,1)$ from the call to graded agreement, the graded agreement property ensures that each correct $v$ outputs $(z_w,g_v)$ from this call for some $g_v\in \{0,1\}$. In particular, $\ell=\ell_{\ell}$ sends $z_{\ell}=z_w$ in the third round, so that each correct node outputs $y_v=z_{\ell}=z_w$. The other case is that no correct node outputs $(z_w,1)$ for some $z_w$ from the call to graded agreement. In this case all nodes output $z_{\ell}$.
\end{proof}